\pdfoutput=1
\documentclass{article}

\usepackage{arxiv}
\usepackage[T1]{fontenc}
\usepackage[utf8]{inputenc}
\usepackage{url}
\usepackage{amsfonts}
\usepackage{hyperref}       % hyperlinks
\usepackage{booktabs}       % professional-quality tables
\usepackage{nicefrac}       % compact symbols for 1/2, etc.
\usepackage{microtype}      % microtypography

\usepackage{amsmath}
\usepackage{graphicx}
\usepackage[export]{adjustbox}
\usepackage{amsthm}
\usepackage{multicol}

\newtheorem{proposition}{Proposition}
\newtheorem{lemma}{Lemma}
\DeclareMathOperator*{\E}{\mathbb{E}}
\newcommand{\Cov}{\mathrm{Cov}}

\newcommand*\samethanks[1][\value{footnote}]{\footnotemark[#1]}

\title{An Improved Tobit Kalman Filter with Adaptive Censoring Limits}

\author
{
	Kostas~Loumponias
	\thanks{K. Loumponias and G. Tsaklidis are with the Department of Mathematics, Aristotle University of Thessaloniki, GR-54124, Thessaloniki, Greece (e-mail: kostikasl@math.auth.gr; tsaklidi@math.auth.gr) }
	\And
	Nicholas~Vretos 
	\thanks{N. Vretos and P. Daras are with the Information Technologies Institute, Centre for Research
		and Technology - Hellas, 6th km Charilaou - Thermi, GR-57001, Thessaloniki,
		Greece, (e-mail:vretos@iti.gr; daras@iti.gr) }
	\And
	George~Tsaklidis \samethanks[1]
	\And
	Petros~Daras \samethanks[2]	
}

\begin{document}
	
\maketitle

\begin{abstract}
This paper deals with the Tobit Kalman filtering (TKF) process when the measurements are correlated and censored.  The case of interval censoring, i.e., the case of measurements which belong to some interval with given censoring limits, is considered. Two improvements of the standard TKF process are proposed,  in order to estimate the hidden state vectors. Firstly, the exact covariance matrix of the censored measurements is calculated by taking into account the censoring limits. Secondly, the probability of a latent (normally distributed) measurement to belong in or out of the uncensored region is calculated by taking into account the Kalman residual. The designed algorithm is tested using both synthetic and real data sets. The real data set includes human skeleton joints’ coordinates captured by the Microsoft Kinect II sensor. In order to cope with certain real-life situations that cause problems in human skeleton tracking, such as (self)-occlusions, closely interacting persons etc., adaptive censoring limits are used in the proposed TKF process. Experiments show that the proposed method outperforms other filtering processes in minimizing the overall Root Mean Square Error (RMSE) for synthetic and real data sets. 

\end{abstract}


\section{Introduction}

Human skeleton motion tracking has been studied for several decades and remains a highly active research field due to its importance in several diverse domains like surveillance applications, medical applications, serious games, educational applications, high performance sports monitoring and others \cite{zhang2013activity}-\nocite{li2015adaptive}\nocite{6626306}\cite{7467541}. With the advent of commercial RGB-D sensors \cite{destelle2014low}, \cite{moon2016multiple},  human skeleton motion tracking has attracted a lot of attention due to the capacity of the sensors to reliably track skeletal joints. However, regardless of the significant progress that has been achieved in both sensors' developement and human skeleton motion tracking research, many applications require more accurate tracking of the human skeleton position and motion. On the sensors' side, high performing sensors (such as the Vicon System), which are able to accurately track at high rates, are very expensive and cumbersome to deploy. On the other hand, affordable, commercial RGB-D solutions (i.e., the Microsoft Kinect, the Xtion Pro and others) often produce low quality human skeleton motion tracking due to their inherent problems (low sampling frequency, moderate depth resolution, UV light interferences, etc.) and also due to their simplistic setup (usually only one such sensor is deployed, resulting in occluding areas and human self-occlusion). 

To overcome these issues and provide an affordable  and, at the same time, reliable solution to the human skeleton motion tracking task, research has been steered towards two general categories of methods: methods that exploit multiple RGB-D sensors \cite{moon2016multiple}, \cite{asteriadis2013estimating} and methods that use various filters able to improve and smooth the sensors' measurements \cite{berti2012kalman}-\nocite{du2015hierarchical}\cite{edwards2014low}.  For the first category, we are confronted with two major flaws: 1) the increase of the cost for monitoring, capturing and processing, and 2) the interferences between devices, which add more noise and restrictions to the problem at hand, thus, making it harder to solve. For the latter, the main drawback is the lack of a framework able to provide reliable estimations of the human skeleton joints. 

In this paper, we introduce a new method, which belongs to the second category of methods. We improve the human skeleton motion tracking by smoothing the Kinect skeleton joints' measurements through a novel Kalman type filtering method adapted to restrictive conditions concerning human skeleton movements. The measurements that we correct and smooth are the 25 Kinect's V2 skeletal joints, which are time series of 3D spatial coordinates  in a 3D space centred in the physical centre of the Kinect's infrared sensor. 

In the literature, in order to smooth spatial coordinates (or a signal), various filters, e.g. Kalman Filter (KF) \cite{kalman1960new}, \cite{lim2000kalman}, Savitzky-Golay filter (SGF) \cite{savitzky1964smoothing}, Particle Filtering \cite{arulampalam2002tutorial} and others have been proposed. One of the most common filters for signal smoothing is KF under the assumption that the singal's measurements are normally distributed. However, KF performs a poor smoothing when the noisy signal contains some extreme measurements (outliers). Then, the hypothesis of normally distributed measurements turns out to be inappropriate. In the case where certain bounds of the denoised signal's values are considered, we can deal with the extreme measurements by providing this information in the KF process.  In order to deal with that, we introduce the censored normal distribution in the KF estimation procedure \cite{moore1956estimation}, \cite{hampshire1992tobit}. The use of censored probabilities theory in data filtering was firstly introduced in \cite{allik2014tobit}, where the Tobit Kalman Filter (TKF) was proposed aiming to estimate an unknown state vector, $ \textbf{x} $, when censored measurements, $\textbf{y}$, are present. In our previous works \cite{loumponias2016},\cite{loumponias2016using}, TKF was utilized in order to filter spatial coordinates of human skeleton, however, no proofs for the TKF process were provided.

In this paper, we propose a new filter, the so-called Adaptive Tobit Kalman Filter (ATKF), which considers an occluded or self-occluded Kinect's skeletal joint as a censored measurement. Our work takes advantage of the approaches presented in  \cite{allik2014tobit}-\cite{loumponias2016using} and proposes a new proof. The proposed approach results in a more accurate estimation of the probability of a measurement to fall into the censoring region and as a consequence, it leads to a more accurate estimation of the state. The proposed ATKF also adapts its censoring region at each time step by considering previous states. The main contributions of this paper are:
\begin{enumerate}
	\item A proof for accurately calculating the covariance matrix of the censored measurements in Tobit Kalman filtering, by incorporating the censoring limits into the equation of censored covariance.
	\item A proof for accurately calculating the probabilities of a latent measurement, $ y^*$, to belong in or out of the uncensored region, by taking into consideration the Kalman residual. 
	\item A new Adaptive Tobit Kalman Filter able to adapt the censoring limits at each time step.
	\item As an application of contributions 1,2 and 3, a new method, which improves the human skeleton tracking in real-time applications is provided. 
	\item A new evaluation metric for human skeleton motion filtering to measure the performance of a filtering technique, when no ground truth data are available.
\end{enumerate}

The rest of the paper is organised as follows. In Section 2, related works are described, while in Section 3, the proposed Adaptive Tobit Kalman Filter is presented in detail. In Section 4, experimental results are drawn, using artificial data as well as real human skeleton motion tracking data. Finally, Section 5, concludes the paper.

\section{Related Work}

Many approaches exist for filtering and motion tracking of the human skeleton  either from images, videos or depth information. We  mention only methods that are most relevant to our work (based on data filtering). For a more detailed discussion we refer to the books \cite{grewal2011kalman} and \cite{moeslund2011visual} for data filtering and human skeleton motion, respectively. 

Similar to our work, Microsoft \cite{skeljoin} proposed various filters for smoothing human skeleton motion data from Kinect devices. Two of them are the simple and the exponential moving average \cite{wei1994time}, \cite{savitzky1964smoothing}, but there is not any reference on how the time windows and the weights should be chosen, since these are application dependent. Edwards \textit{et al.} \cite{edwards2014low} smoothed human skeleton motion data (obtained by a Kinect V2 sensor) using four different filters: 1) the moving average, 2) KF, 3) the Holt double exponential filter \cite{kalekar2004time} and 4) their proposed filter, consisting of a Kalman filter with a Wiener Process Acceleration (WPA) \cite{whitmore1997modelling}. Both the averaging filter and KF had a good smoothing performance but they introduced relatively large amounts of latency, while the other two had good performance and low latency. Finally, the WPA Kalman filter exhibited the best overall performance.

Regarding the filtering process {\it per se}  the most known and well established filtering method is the Kalman Filter (KF). In order to overcome several drawbacks of KF (mainly due to its linear nature), the Extended Kalman Filter (EKF) was proposed in \cite{la1996design}. Although EKF is not an optimal estimator as its linear counterpart, it has been proved that it performs better than KF in terms of smoothing and correcting signals in problems that are non-linear, as is the case in most of the real-life problems. However, EKF tends to be unstable in many applications due to its local nature, leading to incorrect smoothing of a signal that exhibits a high degree of non-linearities. To overcome these problems, the Unscented Kalman Filter (UKF) was proposed in \cite{larsen2011unscented},\cite{gustafsson2012some}. UKF uses a deterministic sampling technique known as unscented transform \cite{julier2002scaled} to gather a minimal set of points around a local mean. By doing so, it provides better results than EKF when the predict and the update functions are highly non-linear and EKF has typically poor performance. Finally, a very successful method is the particle filtering \cite{deutscher2000articulated}, which is a Monte Carlo based filtering method. Though particle filtering is generally very adaptable, it requires a high computational burden, making it practically unsuitable for many real-time applications. 

In the area of censored statistics, all the above mentioned methods have their drawbacks. In Allik \cite{allik2014tobit}, it is stated that the formulation of a standard KF, as an estimator for censored data,  results in a biased estimation of the unknown state. EKF suffers from an undefined Jacobian at the censored region, resulting in an ill-posed Jacobian and thus exhibiting poor performance. On the other hand, UKF  is a less computationally expensive approach, however it is proven to be non-robust when the measurements are close to the censored region \cite{allik2014tobit}. Furthermore, while particle filtering is suitable for estimating the state values when the measurements are censored in certain cases, it has a substantial computational cost. Finally, TKF provides unbiased, recursive estimates of the latent state variables in/near the uncensored regions. TKF is completely recursive and computationally inexpensive, making it a perfect candidate for real-time applications such as the human skeleton motion tracking. Nevertheless, TKF neither takes into account the censored area in calculating the censored measurements variance nor it adapts the limits of the censored area \cite{allik2016tobit}. 

Fei Han \textit{et al.} \cite{han2018improved} concerned TKF for a class of linear discrete-time system with random parameters. The elements of the state space matrices  are allowed to be random variables in order to reflect the reality. Furthermore, they established a novel weighting covariance formula to address the quadratic terms associated with the random matrices. Although, their proposed method with only one censoring limit is coped.

In the area of human skeleton motion tracking, several methods have been proposed involving multiple RGB-D sensors, increasing the complexity and the cost of the solution as mentioned before. In \cite{moon2016multiple}, Sungphil \textit{et al.}, proposed a new method for human skeleton motion tracking using multiple Kinect V1 sensors. They determined the reliability of each 3D joint's position, by combining multiple observations based on Kinect measurements confidence (a value gathered from the sensor). They used the variances of measurements noise in order to identify the contribution of an observation (i.e., a weight) to create a series of fused measurements. Furthermore, they explained how to estimate this variance for each joint through KF. Finally, they presented the average 3D position error of ten activities produced by: 1)  their method, 2) a single Kinect and 3) a simple-average. In all activities  but one (running), their method appeared to give better results than other methods compared with other methods provided in the paper.

Finally, it is worth to mention works on activity recognition that use human skeleton motion filtering as an initial step. In \cite{du2015hierarchical}, \cite{zhu2016co}, a simple SGF is used in order to correct the data. This is achieved through a convolution process by fitting successive subsets of adjacent observations with a low-degree polynomial in a least squares sense \cite{gorry1990general}. Amor \textit{et al.} \cite{amor2016action} dealt with human activity recognition as well, achieving state-of-the-art classification results by using RGB-D sensors. They represented human body as dynamical skeletons and they studied the evolution of the skeletons' shapes as trajectories on manifolds. They performed a median filtering in the temporal dimension in order to de-noise the skeletons' trajectories before using their proposed method.

\section{Proposed Method}

In this section, we briefly describe the censoring data theory and the well-known TKF \cite{allik2014tobit} in order to better highlight the proposed contributions. Then, we demonstrate an alternative approach to the classical TKF, where the update function is generated by taking into account the censoring limits in the measurements covariance matrix calculation, and thus, resulting in a more accurate evaluation of the censored measurements. Finally, we introduce ATKF for human skeleton motion tracking, where the censored region limits (boundaries) are not constant, as is the case in the standard TKF.
\subsection{ Censored and Truncated Data }
Censoring is a condition in which the value of a measurement or observation is only partially known \cite{turnbull1976empirical}. Censoring occurs when a value falls outside the range of a measuring instrument. For example, a bathroom scale might only measure up to 140 kg. If an 150 kg individual is weighed using that scale, the observer would only know that the individual's weight is at least 140 kg (partially known). Censoring should not be confused with the related idea of  truncation; while by censoring, observations result either in knowing the exact value that applies or in knowing that the value lies into an interval, in the truncation case, only observations in a given range are considered by ignoring all the others. Different types of censoring exist \cite{miller2011survival}:
\begin{itemize}
	\item{\textbf{Left censoring}: a data point is below a certain value but it is unknown by how much.}
	\item{\textbf{Interval censoring}: a data point is somewhere on an interval between two values.}
	\item{\textbf{Right censoring}: a data point is above a certain value but it is unknown by how much.}
	\item{\textbf{Type I censoring}: occurs if an experiment has a set number of subjects or items and stops the experiment at a predetermined time, at which point any subjects remaining are right-censored.}
	\item{\textbf{Type II censoring}: occurs if an experiment has a set number of subjects or items and stops the experiment when a predetermined number are observed to have failed; the remaining subjects are then right-censored.}
	\item{\textbf{Random (or non-informative) censoring}: each subject has a censoring time that is statistically independent of its failure time.}
\end{itemize}
In real-life problems, censored data are very frequent and to the best of our knowledge the concept of censoring in human skeleton motion tracking has never been used before.
\subsection{Tobit Kalman Filters}
Tobit Kalman filters \cite{allik2014tobit}, \cite{allik2016tobit}, \cite{tobin1958estimation},  provide a classification scheme for all aforementioned types of censoring. In the case of scalar measurements, the Tobit model is called censored regression model and is  characterised by the stochastic difference non-linear equation
\[ y^*_k=hx_k+v_k,  \]
\begin{equation}
y_k=\begin{cases} y^*_k,&a<y^*_k<b\\
a,&y^*_k \leq a\\
b,&y^*_k \geq b,
\label{eq:tobit}
\end{cases}
\end{equation}
where $ y_k $, $y^*_k$  stand for the censored measurement and the latent variable, respectively, $ h $ is a multiplicative scalar and $a$, $b$ are the lower and  the upper limits of the uncensored region, respectively. The random variable $ v_k $ is drawn from a Gaussian distribution with mean $ 0 $ and variance $ \sigma^2_v $. From (\ref{eq:tobit}), it is obvious that the TKF process is a non-linear one, since when the latent measurement $y^*_k $ falls outside the uncensored region, the censored measurement $ y_k $ does not depend on the variable $ x_k $.

As has been already stated, KF does not provide optimal or unbiased estimates for the states when the measurements are in the censored region. This is due to the fact that the assumptions of KF \cite{harvey1990forecasting} are not met when the noise measurements are censored.

The scalar case can be easily extended to the general case TKF, which is defined as,

\[\textbf{x}_{k+1}=\textbf{A}\textbf{x}_k+\textbf{w}_k,\] \[\quad\textbf{y}^*_k=\textbf{H}\textbf{x}_k+\textbf{v}_k,\]
\begin{equation}
\label{eq:Tobit}
y_{k,i}=\begin{cases} y^*_{k,i},&a_i<y^*_{k,i}<b_{i}\\
a_i&y^*_{k,i} \leq a_i\\
b_i,&y^*_{k,i} \geq b_i,
\end{cases}
i=1,2,...,m\\
\in\mathbb{N},
\end{equation}
where $k$ stands for the discrete time step and $ \textbf{w}_k$ and $ \textbf{v}_k$ are random vector variables following $ N(\textbf{0},\textbf{Q}_k) $ and $ N(\textbf{0},\textbf{R}_k) $, respectively, where  $N(\boldsymbol{\mu},\mathbf{\Sigma})$ denotes the normal distribution with mean $\boldsymbol{\mu}$ and covariance matrix $\mathbf{\Sigma}$. $ \textbf{A}$ and $\textbf{H} $ are the transition and the observation matrices, respectively, while $\textbf{y}_k=(y_{k,i})_{i=1}^{m}, \textbf{y}^*_k=(y^*_{k,i})_{i=1}^{m} $  are the saturated observations (that are Left and Right censoring at the same time), and the latent observations, respectively. Finally, $ m $ designates the dimensionality of the process (which is three in the case of 3D human skeleton motion data). The predict and the update functions of TKF for saturated measurements are described in detail in \cite{allik2014tobit}.

\subsection{Censored Moments}

In this section, we calculate the first, the second moment of a censored measurement $ \textbf{y}=\{y_i\}_{i=1,...,m} $ (no truncated) and the covariance. For that purpose the following Proposition is needed \cite{bg2009moments}: 

\begin{proposition} \label{truncated}
	If the random variable $ \textbf{y}^*=\{y^*_i\}_{i=1,...,m}$ follows a $ m-D $ normal distribution with density function $ f( \textbf{y}^*) $, mean value $ \boldsymbol{\mu}=\{\mu_i\}_{i,...,m}$ and non-singular covariance matrix $ \boldsymbol{\Sigma} = (\sigma_{i,j})_{i,j=1,...,m},$ then, the expected values of $ y^*_{i}$ and $y^*_i\cdot y^*_j $	given that $a_k<y^*_k<b_k$, $ k=1,...,m$, are:
	\begin{equation}
	\E(y^*_i|a_k<y^*_k<b_k,k=1,...,m)=\mu_i + \sum_{k=1}^m \sigma_{i,k}\big(F_k(a_k)-F_k(b_k)\big),\\
	\label{trun_mean}
	\end{equation}	 
	\begin{equation}
	\begin{split}	
	\E(y^*_iy^*_j/a_k<y^*_k<&b_k,k=1,...,n)=\\
	&=\mu_i\cdot \mu_j  + \sigma_{i,j}+\sum_{k=1}^{m}\sigma_{i,k}\frac{\sigma_{j,k}(a_kF_k(a_k)-b_kF_k(b_k))}{\sigma_{k,k}}\\
	&+\sum_{k=1}^{m}\sigma_{i,k}\sum_{q\neq i}\Big(\sigma_{j,q} -\frac{\sigma_{k,q}\sigma_{j,k}}{\sigma_{k,k}}  \Big)\big[\big(F_{k,q}(a_k,a_q)-F{k,q}(a_k,b_q)\big)\\ 
	&-\big(F_{k,q}(b_k,a_q)-F_{k,q}(b_k,b_q)\big) \big] .
	\end{split}
	\label{trun_joint}
	\end{equation}	
\end{proposition}

The functions $ F_i(x) $ and $ F{i,j}(x,y) $ are given by: 
\begin{equation}
\begin{split}
&F_i(x)=\frac{1}{P(a_j<y^*_j<b_j,j=1,...,m)}\cdot
\int_{a_1}^{b_1}...\int_{a_{i-1}}^{b_{i-1}}\int_{a_{i+1}}^{b_{i+1}}...\int_{a_m}^{b_m} f(x,\textbf{y}^*_{-i})\textbf{dy}^*_{-i},
\end{split}
\label{F}
\end{equation}
\begin{equation}
\begin{split}
F{i,j}(x,y)=\frac{1}{P(a_j<y^*_j<b_j,j=1,...,m)} \cdot &
\int_{a_1}^{b_1}..\int_{a_{i-1}}^{b_{i-1}}\int_{a_{i+1}}^{b_{i+1}}..\int_{a_{j-1}}^{b_{j-1}}\int_{a_{j+1}}^{b_{j+1}}..\\&
\int_{a_m}^{b_m} f(x,y,\textbf{y}^*_{-i-j})\textbf{dy}^*_{-i-j}, \\
\label{F2}
\end{split}
\end{equation}
where $ \textbf{y}^*_{-i}=(y^*_1,..,y^*_{i-1},y^*_{i+1},..,y^*_m) $ and $ \textbf{y}^*_{\!\!_{-i-j}}\!\!\!\!\!\!\!=\!(y^*_1,..,y^*_{i-1},y^*_{i+1},..,y^*_{j-1},y^*_{j+1},..,y^*_m). $\\
Next, the following Lemma is provided in order to calculate the censored moments.
\begin{lemma}
	Let $ X $ be a continuous random variable on a probability space $ \Omega $ and $ Z $ a discrete random variable with outcomes $(z_i)_{i=1}^{n}  $. Then, the expected value of the joint probability function $ f_{X,Z}(x,z) $ can be given by \[\E(X,Z)=\sum_{i=1}^{n}z_iE(X|Z=z_i)P(Z=z_i) . \]
	\label{joint_mean}	
\end{lemma}

\begin{proof} 
	We have that
	\[
	\begin{split}
	\E(X,Z)&= \\
	&=\int_{\Omega}\sum_{i=1}^{n}z_i x f_{X,Z}(x,z)dx\\
	&=\int_{\Omega}\sum_{i=1}^{n}z_i x f_{X|Z}(x|z_i)P(Z=z_i)dx\\
	&=\sum_{i=1}^{n}z_i P(Z=z_i)\int_{\Omega}  x f_{X|Z}(x|z_i)dx\\
	&=\sum_{i=1}^{n}z_i P(Z=z_i)\E(X|Z=z_i).\\
	\end{split}
	\]	
\end{proof}
Now the following Proposition can be proved (see Appendix A) using Lemma \ref{joint_mean} and Proposition \ref{truncated}:
\begin{proposition}
	The mean value of the censored variable $ y_i $  with censoring limits $ a_i $ and $ b_i $ (\ref{eq:tobit}), depends only on the censoring limits $a_i, b_i$ and can be written as:
	
	\begin{equation}
	\begin{split}
	E(y_i) = \mu_iP(a_i<y^*_i<b_i)  + \sigma_{i,i}(f_i(a_i)-f_i(b_i)) + a_iP(y^*_i\leq a_i) + b_iP(y^*_i \geq b_i).  
	\end{split}
	\label{cen_mean}
	\end{equation}
	
\end{proposition}   

Furthermore, it can be proved (see Appendix B) that: 
\begin{proposition}
	The variance and the joint mean value of the censored variable  $ y_i $  (\ref{eq:tobit}) and $ y_i, y_j $, respectively, depend only on the censoring limits $\{a_i, b_i\}$ and $ \{a_i, b_i, a_j, b_j\} $, respectively, and are given by:
	\begin{equation}
	\begin{split}
	Var(y_i) &= \mu_i^2(1-P_{un,i})P_{un,i} + \sigma_{i,i}P_{un,i} + a^2_i(1-P_{a,i})P_{a,i}  \\
	&+ b^2_i(1-P_{b,i})P_{b,i}-2a_ib_iP_{a,i}P_{b,i} - \sigma_{i,i}^2(f(a_i)-f(b_i))\\ 
	& +2\mu_i\sigma_{i,i}(f_i(a_i)-f(b_i))(1-P_{un,i})\\
	&+\sigma_{i,i}\big((a_i-\mu_i)f_i(a_i)- (b_i-\mu_i)f_i(b_i)\big) \\
	&-2\Big( \mu_iP_{un,i} + \sigma_{i,i}\big(f_i(a_i)-f(b_i)\big)\Big)\Big(a_iP_{a,i} + b_iP_{b,i} \Big)\\ 
	\end{split}
	\label{cen_var}
	\end{equation}
	and	
	\begin{equation}
	\begin{split}
	\E(y_{i}y_{j})&=a_ib_j P(1) +b_ib_j P(3) + a_ia_jP(7) + b_ia_jP(9)\\
	& + b_j\E(y^*_{i}|a_i< y^*_{i} <b_i, y^*_j \geq b_j)P(2)\\
	& + a_i\E(y^*_{j}|a_j< y^*_{j} <b_j, y^*_i \leq a_i)P(4)\\
	& +\E(y^*_{i}y^*_{j}|a_i< y^*_{i} < b_i, a_j < y^*_{j} <b_j)P(5)\\
	& + b_i\E(y^*_{j}|a_j< y^*_{j} <b_j, y^*_i \geq b_i)P(6)\\
	& + a_j\E(y^*_{i}|a_i< y^*_{i} <b_i, y^*_j \leq a_j)P(8).\\
	\end{split}
	\label{cen_joint}
	\end{equation}
	\label{Censored}		
\end{proposition} 
The probabilities $ P_{un,i}, P_{a,i}, P_{b,i} $ and $ P(j)_{j=1,...,9} $ are defined as follows:
\[ 
\begin{split}
&P_{un,i} = P(a_i<y^*_i<b_i), P_{a,i}= P(y^*_i \leq a_i),\\
& P_{b,i} = P(y^*_i \geq b_i), P(1)= P(y^*_i\leq a_i, y^*_j\geq b_j),\\
& P(2)= P(a_i <y^*_i< b_i, y^*_j\geq b_j),\\
&P(3)= P(y^*_i\geq b_i, y^*_j\geq b_j),\\ 
&P(4)= P(y^*_i\leq a_i, a_j<y^*_j<b_j),\\
&P(5)= P(a_i<y^*_i<b_i, a_j<y^*_j< b_j),\\
&P(6)= P(y^*_i\geq b_i, a_j<y^*_j< b_j), \\
&P(7)= P(y^*_i\leq a_i, y^*_j\leq a_j),\\  
& P(8)= P(a_i<y^*_i< b_i, y^*_j\leq a_j),\\
& P(9)= P(y^*_i\geq b_i, y^*_j\leq a_j).
\end{split} 
\]
The truncated expected values $ \E(y^*_{i}|\cdot), E(y^*_{i}y^*_{j}|\cdot) $ in (\ref{cen_joint}) are calculated in Appendix B. Hence, the covariance matrix of the censored variable $ \textbf{y}$ can be calculated by (\ref{cen_mean})-(\ref{cen_joint}).

\subsection{Corrected Tobit Kalman Filter}

In this paper as in \cite{han2018improved},\cite{allik2014estimation}, we calculate the a posteriori estimation, $\hat{\textbf{x}}_k $, as a linear combination of the a priori estimation, $ \hat{\textbf{x}}^-_{k}$, and the censored measurement $ \textbf{y}_k $. Although these estimations are not optimal, it is proved that they minimize the trace of state error covariance \cite{masreliez1977robust}. Next, we provide the predict and update function of the proposed $ TKF $.\\
\underline{The Predict function}:
\begin{equation}
\label{eq:first}
\hat{\textbf{x}}^-_{k} = \textbf{A}\hat{\textbf{x}}_{k-1}, \qquad \qquad
\end{equation}
\begin{equation}
\textbf{P}^-_{k} = \textbf{A}\textbf{P}_{k-1}\textbf{A}^T+\textbf{Q}_k.
\end{equation}
\underline{The Update function}:
\begin{equation}
\; \;\textbf{R}_{k,1}=\E \big( (\textbf{x}_{k}-\hat{\textbf{x}}^-_{k})(\textbf{y}_k-\E(\textbf{y}_k))^T | \textbf{y}_{k-1} \big), 
\end{equation}
\begin{equation}
\qquad \textbf{R}_{k,2}=\E \big( (\textbf{y}_k-\E(\textbf{y}_k))(\textbf{y}_k-\E(\textbf{y}_k))^T| \textbf{y}_{k-1} \big),
\label{Cen_Cov}
\end{equation}
\begin{equation}
\label{eq:Gain}
\textbf{K}_k = \textbf{R}_{k,1}\textbf{R}_{k,2}^{-1}, \qquad  \qquad  \qquad \qquad  \quad
\end{equation}
\begin{equation}
\label{eq:Aver}
\hat{\textbf{x}}_k=\hat{\textbf{x}}^-_k+\textbf{K}_k(\textbf{y}_k-\E(\textbf{y}_k)), \qquad  \quad \;
\end{equation}
\begin{equation}
\label{eq:final}
\textbf{P}_k=\textbf{P}^-_k-\textbf{K}_k\textbf{R}_{k,1}^T. \qquad  \qquad \qquad  \; \; \,
\end{equation}
The predict function is the same as in case of standard KF \cite{kalman1960new}, since, the censored measurements are not used in this stage. Matrix $ \textbf{R}_{k,1} $ has been calculated in \cite{allik2014estimation} and takes the form 
\begin{equation}
\textbf{R}_{k,1}=\textbf{P}^-_{k}\textbf{H}^T\textbf{D}_{un,k},
\end{equation}
where $ \textbf{D}_{un,k} $ is a $m \times m$ diagonal matrix, and its entries are the probabilities of a measurement to be uncensored,  at time step $ k $. More specifically, the $ i^{th}$ diagonal element of $ \textbf{D}_{un,k} $, is the probability that a latent measurement $ y^*_{k,i} $  belongs to the uncensored region. Furthermore, we denote by $ \textbf{D}_{\textbf{a},k} $, $ \textbf{D}_{\textbf{b},k} $ the diagonal matrices, where its entries are the probabilities of a measurement to be censored from below or above, respectively, at time step $ k $. It is proved (see Appendix C) that:
\begin{equation}
\label{eq:Dunk}
\textbf{D}_{un,k}=diag
\begin{bmatrix}
\Phi(b_{k,1})-\Phi(a_{k,1})\\
... \\
\Phi(b_{k,m})-\Phi(a_{k,m})
\end{bmatrix},	
\end{equation}

\begin{equation}
\label{eq:Dbelow}
\textbf{D}_{\textbf{a},k} =diag
\begin{bmatrix}
\Phi(a_{k,1})\\
... \\
\Phi((a_{k,m})
\end{bmatrix},
\end{equation}

\begin{equation}
\label{eq:Dabove}
\textbf{D}_{\textbf{b},k} =diag
\begin{bmatrix}
1-\Phi(b_{k,1})\\
... \\
1-\Phi(b_{k,m})
\end{bmatrix}, 
\end{equation}
where $\Phi$ stands for the cumulative function of $N(0,1)$. In \cite{allik2014tobit}, ${b_{k,i}}$ and ${a_{k,i}}$ are calculated as (we denoted them with $ * $ to not confuse them with the proposed)
\begin{eqnarray}
{b^*_{k,i}} &=& \frac{b_{i}- m_{k,i}}{\sqrt{r_{(i,i),k}}} \\
{a^*_{k,i}} &=& \frac{a_{i}- m_{k,i}}{\sqrt{r_{(i,i),k}}} 
\end{eqnarray}
where $\textbf{R}_k=(r_{(i,j),k})$, $\textbf{m}_k = \textbf{H}\hat{\textbf{x}}^-_k $ and $ \textbf{S}_k=\textbf{H}\textbf{P}^-_{k}\textbf{H} + \textbf{R}_k = (s_{(i,j),k})$. We notice that the information from the Kalman residual, $ (\textbf{y}^*_k - \textbf{m}_k) $,  is omitted. In our case (see Appendix C) these amounts are as follows:
\begin{eqnarray}
{b_{k,i}} &=& \frac{b_{i}- m_{k,i}}{\sqrt {s_{(i,i),k}}} \label{eqn:Tu}\\
{a_{k,i}} &=& \frac{a_{i}- m_{k,i}}{\sqrt {s_{(i,i),k}}} \label{eqn:Tl}
\end{eqnarray}
In (\ref{eqn:Tu}) and (\ref{eqn:Tl}), as opposed to \cite{allik2014tobit}, we have incorporated in the denominator the term $(\textbf{H}\textbf{P}^-_{k}\textbf{H})_{i,i} $, which consequently, adds information into (\ref{eqn:Tu}) and (\ref{eqn:Tl}), concerning the Kalman residual. By doing so, the probability of a measurement to belong to the uncensored region is estimated more accurately.

The mean vector of the censored measurement $ \textbf{y}_k $ given the previous censored measurement $\textbf{y}_{k-1} $, can be written (in matrix notation) using (\ref{cen_mean}) as:
\begin{equation}
\begin{split}
&\E(\textbf{y}_k|\textbf{y}_{k-1}) =\textbf{m}_k\!\cdot\!\textbf{D}_{un,k} +\textbf{S}_k
\cdot diag(f_i(a_i)-f_i(b_i))_{i=1,..,m}+ \textbf{a}\!\cdot\!\textbf{D}_{\textbf{a},k} + \textbf{b}\!\cdot\!\textbf{D}_{\textbf{b},k}.  
\end{split}
\label{tkf_mean}
\end{equation}
The covariance matrix, $ \textbf{R}_{k,2} $, of the censored measurement, $ \textbf{y}_k $, given the last censored measurement, $ \textbf{y}_{k-1} $, is calculated via Proposition \ref{Censored}. In particular, the diagonal elements, $ Var(y_{(i,i),k}|\textbf{y}_{k-1}) $, of  $ \textbf{R}_{k,2} $ are calculated as $ Var(y_i) $ (\ref{cen_var}), where the mean vector, $  \boldsymbol{\mu} $, and covariance matrix, $  \boldsymbol{\Sigma} $, in our proposed model are equal with  $ \textbf{H}\hat{\textbf{x}}^-_k $ and $ \textbf{H}\textbf{P}^-_{k}\textbf{H} + \textbf{R}_k $, respectively, and the probabilities $P_{un,i} , P_{a,i}, P_{b,i} $ for $ i=1,...,m $ are given in $ (\ref{eq:Dunk})-(\ref{eq:Dabove}) $. In the same way, the off-diagonal elements, $ E(y_{k,i}y_{k,j}|\textbf{y}_{k}) -  E(y_{k,i}|\textbf{y}_{k})E(y_{k,j}|\textbf{y}_{k})$, of  $ \textbf{R}_{k,2} $ are calculated as $ E(y_{i}y_{j}) $ (\ref{cen_joint}). 

In what follows we denote by TKF$^c$ the filter described through (\ref{eq:first})-(\ref{eq:final}) and by TKF the filter described in \cite{allik2014tobit}, \cite{allik2014estimation}. In \cite{allik2014estimation},  the covariance matrix, $\textbf{R}^*_{k,2} $, of the censored measurement $ \textbf{y}_k $, is given by
\begin{equation}
\label{varbeth}
\textbf{R}^*_{k,2}=\textbf{D}_{un,k}\textbf{H}\textbf{P}_k^-\textbf{H}^T\textbf{D}_{un,k}+\textbf{R}_k^*,
\end{equation}
where  $\textbf{R}_k^*$ is a diagonal matrix, where its entries  are the truncated variances of $ y^*_{k,i} $ for $ i=1,...,m. $ (\ref{eq:Tobit}).

The main difference between (\ref{Cen_Cov}) and (\ref{varbeth}) is that in (\ref{varbeth}) the limits $ a_i $ and $ b_i $ appear only in the matrices  $\textbf{D}_{un,k}, \textbf{D}_{\textbf{a},k}$ and $ \textbf{D}_{\textbf{b},k}$. We notice that if $ a_i=0 $ and $ b_i$ is big enough (that is, only non-negative measurements are considered), then (\ref{varbeth}) provides a satisfactory approximation of the covariance matrix of the censored measurements. In order to clarify the notation and illustrate the difference between $ \textbf{R}_{k,2} $ and $ \textbf{R}^*_{k,2} $, we provide an illustrative example as follows: we examine the censored  covariance matrix for the random multidimensional $ \textbf{Y}^* \sim N(\textbf{m}_k, \textbf{S}_k) $ with censoring limits $ \textbf{a}=(-1, -3, 1)^t $ and $ \textbf{b}=(1, 7, 4)^t $. We define the mean vector, $ \textbf{m}_k $, and the covariance matrix $ \textbf{S}_k $ to be equal with,
\[\textbf{m}_k = \textbf{H}\cdot(2, 2, 3)^T,\] 
and 
\[
\textbf{S}_k=\textbf{H}
\begin{bmatrix}
4&3&4 \\
3&4&4  \\
4&4&4
\end{bmatrix}\textbf{H}^T + \textbf{R}_k,
\]
while, without loss of generality, we define  $ \textbf{H} $ and $ \textbf{R}_k $ to be equal with the $ 3\times3 $ identity matrix. Then, we proceed as follows: 1) we produce $ 10^5 $ random measurements from $ N(\textbf{m}_k, \textbf{S}_k) $ 100 times. 2) Each time, we calculate the sampling covariance matrix derived from the censored measurements. 3) We calculate the arithmetic mean, $ \textbf{R}_{s} $, of the 100 sampling covariance matrices. 4) The covariance matrices $ \textbf{R}_{k,2} $ and $ \textbf{R}^*_{k,2} $ are calculated by (\ref{Cen_Cov}) and (\ref{varbeth}), respectively. As it can been by (\ref{proposed})-(\ref{oldversion}), the proposed covariance matrix, $ \textbf{R}_{k,2} $, is almost identical with the sampling covariance matrix, $ \textbf{R}_{s} $. 

\begin{equation}
\textbf{R}_s=
\begin{bmatrix}
0.4648  &  0.6962  &  0.5083\\
0.6962  &  4.7754  &  1.9195\\
0.5083  &  1.9195  &  1.4384
\end{bmatrix},
\label{proposed}
\end{equation}

\begin{equation}
\textbf{R}_{k,2}=
\begin{bmatrix}
0.4651  &  0.6962  &  0.5085\\
0.6962  &  4.7747  &  1.9189\\
0.5085  &  1.9189  &  1.4379
\end{bmatrix},
\end{equation}

\begin{equation}
\textbf{R}^*_{k,2}=
\begin{bmatrix}
0.2724  &  0.4719  &  0.5151 \\
0.4719  &  5.0000  &  3.2744 \\
0.5151  &  3.2744  &  3.2002
\end{bmatrix}.
\label{oldversion}
\end{equation}

The marginal probability function, $ f(y_{k,i}|\textbf{y}_{k-1}) $, of the $ i^{th} $ component of the censored measurement $ \textbf{y}_{k} $ given the last measurement, $ \textbf{y}_{k-1}, $ is,
\begin{equation}
\begin{split}
f(y_{k,i}|\textbf{y}_{k-1})&=\frac{1}{\sqrt {s_{(i,i),k}}}\phi\bigg(\frac{y_{k,i}-m_{k,i}}{\sqrt {s_{(i,i),k}}}\bigg)u(y_{k,i}-a_i)u(b_i-y_{k,i})\\&+\Phi\bigg(\frac{b_i-m_{k,i}}{\sqrt {s_{(i,i),k}}}\bigg)\delta(a_i-y_{k,i}) \\
&+\Bigg(1-\Phi\bigg(\frac{b_i-m_{k,i}}{\sqrt {s_{(i,i),k}}}\bigg)\Bigg)\delta(b_i-y_{k,i}),
\end{split}
\label{eq:pdf}
\end{equation}
where $ \phi $ and $ \Phi $ are the probability and the cumulative distribution functions of the standard normal distribution, respectively,  $\delta$ is the Kronecker delta function and $ u $ stands for the Heavyside function, where $ u(x)=1 $, when $ x > 0 $ and $ u(x)=0 $, otherwise. 

The next step in our procedure is to calculate the likelihood function by taking into consideration the censored data distribution. The likelihood function for the censored measurements $ (y_{k,i})_{k=1}^{n} $ by (\ref{eq:pdf}), (\ref{eq:Dbelow}) and (\ref{eq:Dabove}) can be calculated as:
\begin{equation}
\begin{split}
L_i(y_{1,i},...,y_{n,i})&={\displaystyle\prod_{y_{k,i}=a_i} \Phi(a_{k,i})}\times {\displaystyle\prod_{y_{k,i}=b_i}(1-\Phi(b_{k,i}))}\\&\times \hspace{-0.5cm} {\displaystyle\prod_{ a_i <~ y_{k,i} <~ b_i} \frac{1}{\sqrt {s_{(i,i),k}}} \phi \Bigg(\frac{y_{k,i}-m_{k,i}}{\sqrt {s_{(i,i),k}} } \Bigg)  },
\label{mleonedimtob}
\end{split}
\end{equation}
In the case that the components of $ \textbf{y}_k $ are mutually independent, the likelihood function of the censored measurements $(\textbf{y}_k)_{k=1}^{n} $ takes the form:
\begin{equation}
L(\textbf{y}_1,...,\textbf{y}_n)= {\displaystyle\prod_{i=1}^{m} L_i(y_{1,i},...,y_{n,i}) }.
\label{eq:mletob}
\end{equation}

In the case of \cite{allik2014tobit}, the likelihood function becomes
\begin{equation}
\begin{split}
L^*_i(y_{1,i},...,y_{n,i})&={\displaystyle\prod_{y_{k,i}=a_i} \Phi(a^*_{k,i})}\times {\displaystyle\prod_{y_{k,i}=b_i}\Big(1-\Phi(b^*_{k,i})\Big)}\\& \times {\displaystyle\prod_{ a_i <~ y_{k,i} <~b_i} \frac{1}{\sqrt{r_{(i,i),k}}} \phi \Bigg(\frac{y_{k,i}-m_{k,i}}{\sqrt{r_{(i,i),k}}} \Bigg)}.
\end{split}
\label{mleonedimallik}
\end{equation}
Note that the denominator does not take into account the specific distribution of the measurements.\\

\subsection{Adaptive Tobit Kalman Filter used to Human Skeleton Tracking}

In what follows, we use the Microsoft Kinect V2 sensor to record 3D point sequences (human skeletons) of a human in motion \cite{kar2010skeletal}. In human skeleton tracking, the body is represented by a number of joints (25 in total), corresponding to different body parts such as head, neck, shoulders, etc (see Fig. \ref{fig:skeleton} \cite{kinect2016}). Each joint is represented by the vector of its Euclidean 3D space coordinates $ [z_1, z_2, z_3] $ and our aim is to denoise the measurements  for every joint in order to improve the representation of human movements. Thus, we denoise each one of the joints' coordinates separately; the input is the vector of the joints' coordinates, $ \textbf{y}^*_k=[y^*_{k,1},y^*_{k,2},y^*_{k,3}] $  (latent measurement), and the output is the vector of the denoised states coordinates, $ \textbf{x}_k=[x_{k,1},x_{k,2},x_{k,3}] $.

\begin{figure}[ht]
	\centering
	\includegraphics[width= 3in]{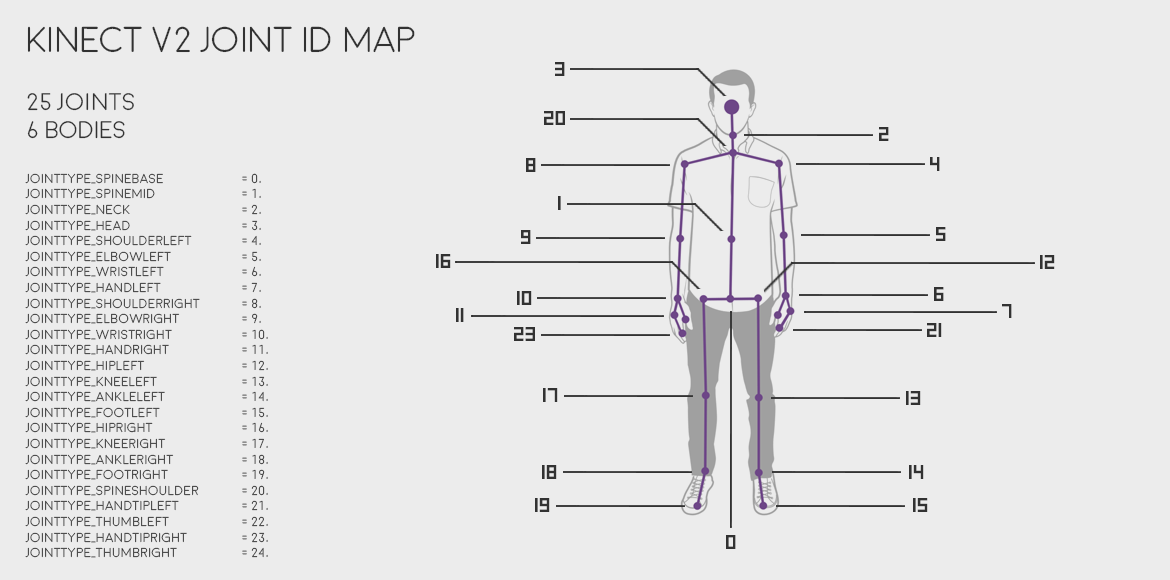}
	\caption{Human skeleton's joints map of the Kinect V2 sensor.}
	\label{fig:skeleton}
\end{figure}

To start tracking, we define the initial observation and the transition matrices to be equal to the identity matrix. Therefore, we define the covariance matrix of the noise measurement, $ \textbf{R},$ to be
\begin{equation}
\textbf{R}=0.01
\begin{bmatrix}
1&0&0 \\
0&1&0  \\
0&0&1
\end{bmatrix}.
\label{eq:Rmat}
\end{equation}
We chose to initialize $ \textbf{R} $ in that way, under the assumption that Kinect exhibits significant errors in human skeleton tracking. To support our claim, we conduct small scale experiments proving that even if a person is at rest and in front of the Kinect, the error in the displacement estimation between measurement and ground truth data is almost 0.02 meters \cite{mobini2014accuracy}, \cite{galna2014accuracy}, thus a variance of 0.01 m$^2 $ seems to be a valid choice.

In KF \cite{berti2012kalman}, \cite{edwards2014low}, no restrictions in joints' movements have been taken into account, as opposed to the proposed method. To that end, in our experiments we have used, beyond the Kinect V2 sensor, the state-of-the-art Vicon tracking system as a ground truth reference. In Vicon data, for various recordings, we observe that the velocity of the spatial coordinates $ z_1 $ and $ z_3 $ did not exceed 34 cm per frame, for every joint, and the  $ z_2 $ coordinate did not exceed 18 cm per frame. In what follows we will use these restrictions in order to correct the data produced by the Kinect sensor. By applying  these restrictions we constructed ATKF  with limits  $ \textbf{l}^1_k$  and   $  \textbf{l}^2_k$  for the vector of the spatial coordinates, $ [y^*_{k,1}, y^*_{k,2}, y^*_{k,3}] $, as follows:
\begin{equation}
\textbf{l}^2_k=\textbf{H}\hat{\textbf{x}}_{k-1} + \textbf{c},
\label{eq:ab_limit}
\end{equation}
and
\begin{equation}
\textbf{l}^1_k=\textbf{H}\hat{\textbf{x}}_{k-1} - \textbf{c},
\label{eq:bl_limit}
\end{equation}
where the observation matrix, $\textbf{H}$, is the identity matrix for smoothing approaches,$ \textbf{l}^1_k$  and   $  \textbf{l}^2_k$ are the limits of ATKF at time $ k $, which depend on the previous estimation of spatial coordinates, $ \hat{\textbf{x}}_{k-1} $, and the vector $ \textbf{c}$, which for human skeleton tracking is experimentally found to be
\begin{equation}
\textbf{c}=(0.34,0.18,0.34).
\end{equation}
Thus, for the latent measurement  $ \textbf{y}^*_k=[y^*_{k,1},y^*_{k,2},y^*_{k,3}] $ at time $ k $ we get
\begin{equation}
y_{k,i}=
\begin{cases} 
y^*_{k,i}, & l^1_{k,i}<y^*_{k,i}<l^2_{k,i}\\
l^1_{k,i}, & y^*_{k,i} \leq l^1_{k,i}\\
l^2_{k,i}, & y^*_{k,i} \geq l^2_{k,i}.
\end{cases}
i=1,2,3.
\end{equation}

This model corrects Kinect measurements, when they have high abnormal velocity. It should be noted that, if $l^1_{k,i}\to-\infty$ and $l^2_{k,i} \to\infty$ (i.e, the range of ATKF limits becomes very large), ATKF tends to the standard KF, because in this case the Kinect measurements belong to the uncensored region and consequently they are known. Due to this fact, we expect in some recordings, which do not include big or fast joints' movements (thus, the Kinect measurements always  belong to the uncensored region), to get almost the same results concerning RMSE for KF as well as for ATKF.

In order to create a general model for smoothing Kinect V2 measurements without having to estimate the matrix $ \textbf{Q}_k $ for every time-window (because this is time consuming), we assume that this matrix is constant. Substituting for $ \textbf{R} $ in the likelihood function (\ref{mleonedimtob}), the covariance matrix of the noise process, $ \textbf{Q} $, can be estimated. By experimenting on various joints' movements, it is derived that the values of $ \textbf{Q} $ are smaller than those of matrix $ \textbf{R} $ and generally they depend on the  speed of the human skeleton's joints. Regarding slow joints' movements, the entries of \textbf{Q} are smaller than $10^{-4}$ and for faster joints' movements they lie  between $ 10^{-3} $ and $ 10^{-2} $. We notice that in some cases, where the entries of $ \textbf{Q} $ appeared to be quite large (in the order of $ 10^{-2} $), the human skeleton moved too quickly in an abnormal manner due to occlusions and/or self-occlusions. Thereafter, we assume that the covariance matrix of the noise process is
\begin{equation}
\textbf{Q}=0.0025
\begin{bmatrix}
1&0&0 \\
0&1&0  \\
0&0&1
\end{bmatrix},
\label{eq:Qmat}
\end{equation}
otherwise, if we define smaller or larger values, ATKF will be either  over-smoothed or will not denoise the Kinect measurements. Therefore, the matrix  $\textbf{Q} $ given in (\ref{eq:Qmat}), seems to be an appropriate choice for smoothing the Kinect V2 sensor measurements of human skeleton tracking.\\

\section{Experiments}

In this section, we conduct three sets of experiments to evaluate TKF$^c$ and ATKF compared to other methods. We use 1) TKF and 2) TKF$^c$ in the first experimental set (oscillator), which is employed in \cite{allik2014estimation}. Next, we use 1) SGF, 2) KF, 3) TKF, 4) TKF$^{c}$ and 5) ATKF in order to smooth data for two different experimental sets: a) Real-life data captured by a Kinect sensor, b) Real-life data captured by both a Kinect sensor and a Vicon system.

\subsection {Oscillator}

In the first experimental set, we present a motivating example of tracking a sinusoidal model by a TKF and TKF$^c$, when the measurements are saturated. Let the state space equations have the form of (\ref{eq:Tobit}), with state space matrices

\begin{equation}
\textbf{A}=c
\begin{bmatrix}
cos(w)&-sin(w)\\
sin(w)&\quad cos(w)
\end{bmatrix},
\label{eq:A}
\end{equation}
and
\begin{equation}
\textbf{H}=
\begin{bmatrix}
1&0
\end{bmatrix},
\label{eq:H}
\end{equation}
where $ c=0.999 $ and $ w=0.005\cdot2\pi$. The disturbance $\textbf{w}_k $ is assumed to be normally distributed, i.e. $\textbf{w}_k\sim N(\textbf{0},\textbf{Q}) $, where\\
\begin{equation}
\textbf{Q}=
\begin{bmatrix}
0.05^2&0\\
0&0.05^2
\end{bmatrix},
\label{eq:Q}
\end{equation}
while, the measurement noise, $ v_k $, is normally distributed, $v_k\sim N(0, 0.5) $. The initial state vector is equal to $ \textbf{x}_0=[5\quad 0]^T $ with covariance matrix $ \textbf{P}_0 = \textbf{I}_{2} $, the  censored limits are  $ a = -0.5 $ and $ b = 0.5 $. Therefore, by the above example we produce censored (saturated) measurements, $y_k$, where $k=1,...,1000$.

Next, we repeat the above process 100 times and we calculate the filters' RMSEs for each iteration. The means of filters's RMSEs for 100 iterations are presented in Table \ref{tab:rmse_x}, where we provide separate RMSEs for the two  estimated coordinates of the state vector, $ \textbf{x}_k$. It can be observed that the corrected TKF$^c$ outperforms TKF in state estimation (Fig. \ref{fig:rmse_x}). This is due to the fact that in TKF some important terms are ignored when calculating $ \textbf{R}_{k,2}^* $ (\ref{varbeth}), while these terms are included in TKF$^c$ process (\ref{Cen_Cov}).   

\begin{table}[ht]
	\renewcommand{\arraystretch}{1.3}
	\begin{center}
		\begin{tabular}{ |c|c|c| }
			\hline
			\textbf{Filter}  & \textbf{Mean RMSE of} $\hat{\textbf{x}}_1 $ &  \textbf{Mean RMSE of} $\hat{\textbf{x}}_2 $ \\
			\hline
			TKF &  0.4434 & 0.5464\\
			TKF$^{c} $ & \textbf{0.4066} & \textbf{0.5192} \\
			\hline
		\end{tabular}
	\end{center}
	\caption{The mean of RMSEs for the filters TKF and TKF$^{c} $, respectively.}
	\label{tab:rmse_x}
\end{table}

\begin{figure}[ht]
	\centering
%	\subfloat[$ \hat{x}_1 $]
	    {
		\includegraphics[width=4in]{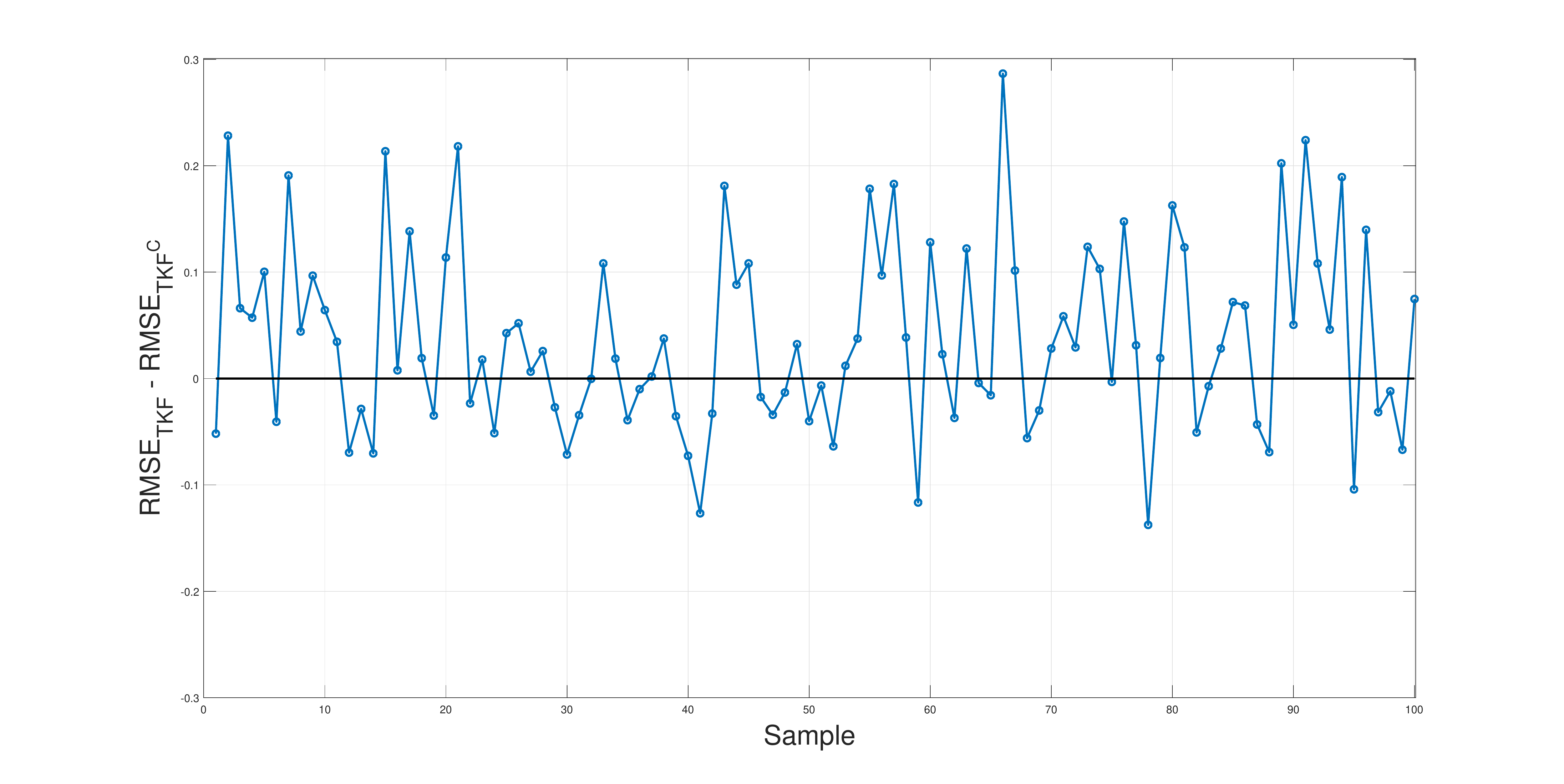}
	    }
    \hfill
%	\subfloat[$ \hat{x}_2 $]
	    {
		\includegraphics[width=4in]{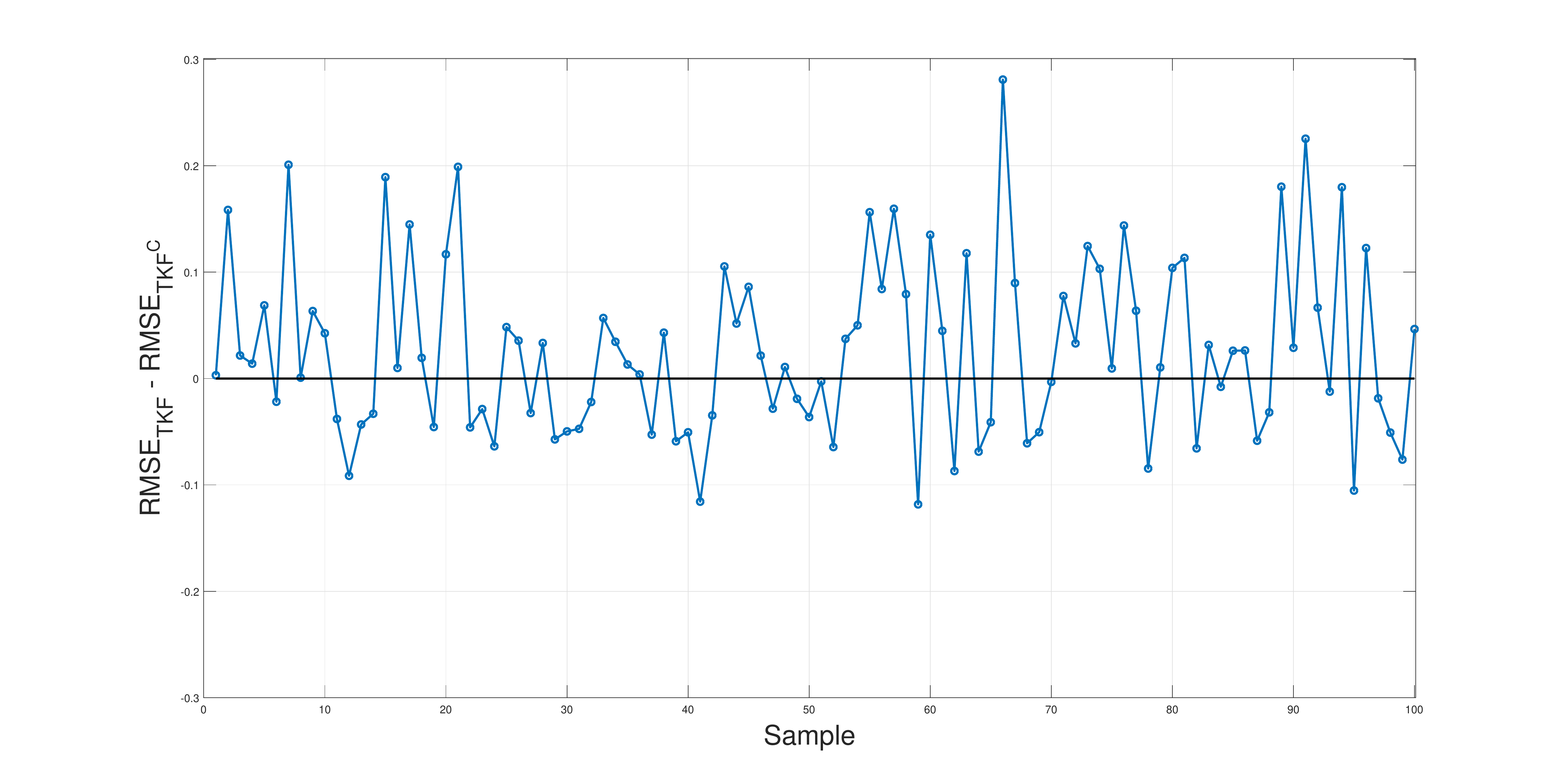}
	    }
	\caption{The difference between TKF's and TKF$^c$'s RMSE for each iteration.}
	\label{fig:rmse_x}
\end{figure}

\subsection{Recordings by the Kinect Sensor}
In the second experiment set, we record various human movements by a single Kinect V2 sensor. In some of the recordings, the human skeleton motion exhibits an important error on the $ z_2 $ axis (practically, the human skeleton seems to "fall down") for one or two frames. We apply the above mentioned filters to correct this specific error. 

In order to evaluate the performance of the different filters, we propose a novel metric $m_i$, to better examine the result of smoothing the joints' movements. Let us denote by $ g_{k,i} $  the filtering of the $ i^{th} $ component of the measurement $ y_{k,i}$ at time $ k $. Then,
\begin{equation}
\label{eq:metric}
m_i=average\Big[(dg_{k,i})^2\Big]_{k=1}^{n-1} \qquad  \qquad \quad
\end{equation}
where $ dg_{k,i}=g_{k+1,i}-g_{k,i}$ and $n$ is the number of measurements.

In the case of TKF$^c$ and TKF we use the device limits. For instance, the ranges of Kinect spatial coordinates  $ z_1, z_2 $ and $ z_3 $ (depth) are approximately $ [-3m,\,3m] $, $ [-1.5m,\,3m] $ (if the Kinect V2 sensor is located $ 1.5m $ over the ground)  and $ [0.5m,\,5m] $, respectively. Thus, we can use these limits for the Kinect measurements in order to test TKF and  TKF$^{c}$. The covariance matrices of TKF$^{c}$ and TKF for the noise measurement, $ \textbf{R}$, are defined as in ATKF (\ref{eq:Rmat}), while the covariance matrices for the noise process, $ \textbf{Q} $, can be estimated using the likelihood functions (\ref{mleonedimallik}) and (\ref{eq:mletob}), respectively. By experimenting on various joints' movements, we get that the entries of $ \textbf{Q} $ are the same as in the case of ATKF, therefore, we can use the same matrix $\textbf{Q} $ given by (\ref{eq:Qmat}). In the case of KF, the covariance matrix, $ \textbf{R}$, is defined as in ATKF (\ref{eq:Rmat}) and the covariance matrix for the noise process, $\textbf{Q}$, is estimated by the log-likelihood function given in \cite{hamilton1994time}. The results showed (in the same experiments as we mentioned before), that the entries of $\textbf{Q}$ are almost the same as in the case of ATKF, thus, the matrix $ \textbf{Q} $ is defined as in (\ref{eq:Qmat}).

In our experiments we take the overall average $ M $ of the metrics $ m_i $ for various recordings. The results showed that ATKF achieves better performance in noise reduction than the other filters (see Table\,\ref{tab:metric} ), especially in the cases where the skeleton seems to collapse, while KF, TKF$^{c} $ and TKF have almost the same overall average $ M $ and SGF has a poor performance. As can be seen in Fig. \ref{fig:cordy} for two different experiments, the head's spatial coordinates $z_2$ of the human skeleton resulted from ATKF, do not (correctly) follow the error produced by the Kinect sensor. It can be seeing (Fig. \ref{fig:cordy})  that although KF, TKF$ ^c $ and TKF improve the human skeleton motion, they provide inferior results than the ones produced by ATKF, while SGF has the worst performance among all. In the first experiment illustrated in Fig. \ref{fig:cordy}a, the ATKF skeleton followed the sharp "fall" for almost 5 cm, while KF, TKF$ ^c $ and TKF skeletons for 12 cm, and the SGF skeleton for 20 cm. The joint based average $m_i$ as opposed to the overall experiments average $M$ of ATKF in this experiment is $ 0.350*10^{-3} $, while in KF, TKF$ ^c $ and TKF is $0.409*10^{-3} $ and in SGF is $0.797*10^{-3}$. 

\begin{table}[ht]
	\renewcommand{\arraystretch}{1.3}
	\begin{center}
		\begin{tabular}{ |c|c| }
			\hline
			\textbf{Filter}  & \textbf{Overall Average $ M $} \\
			\hline
			SGF& $ 0.790*10^{-3} $\\
			KF&  $ 0.436*10^{-3} $ \\
			TKF & $ 0.433*10^{-3} $ \\
			TKF $^{c} $&  $ 0.433*10^{-3} $ \\
			ATKF&$\textbf{0.362}*\textbf{10}^{\textbf{-3}} $ \\
			\hline
			Kinect V2 & $ 1.70*10^{-3} $\\
			\hline
		\end{tabular}
	\end{center}
	\caption{The overall average $ M $ of the recordings for the Kinect V2 sensor and the filters.}
	\label{tab:metric}
\end{table}

\begin{figure}[ht]
	\centering
%	\subfloat[]
	{
		\includegraphics[width=4in]{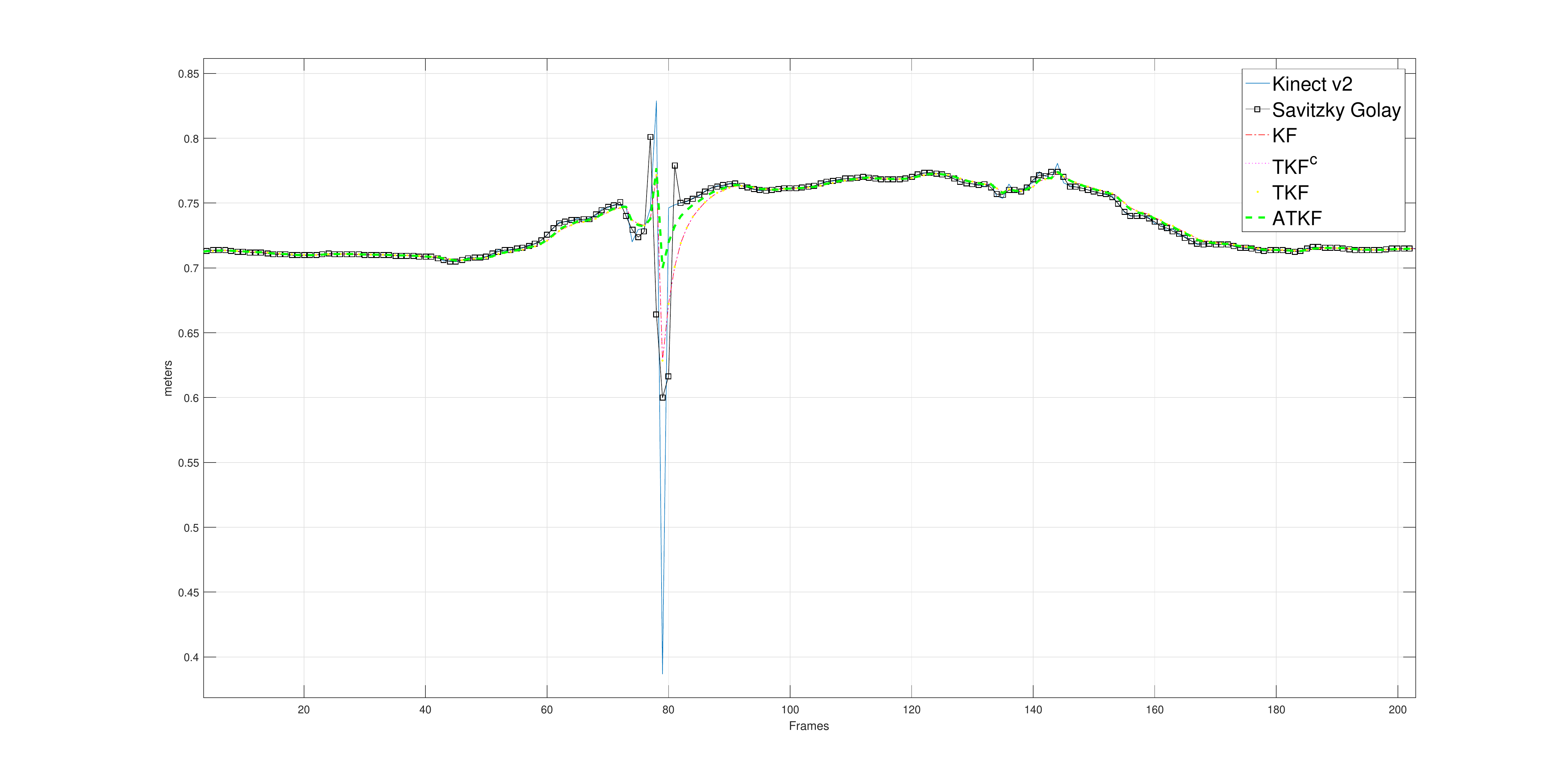}
    }
	\hfill
%	\subfloat[]
	{
		\includegraphics[width=4in]{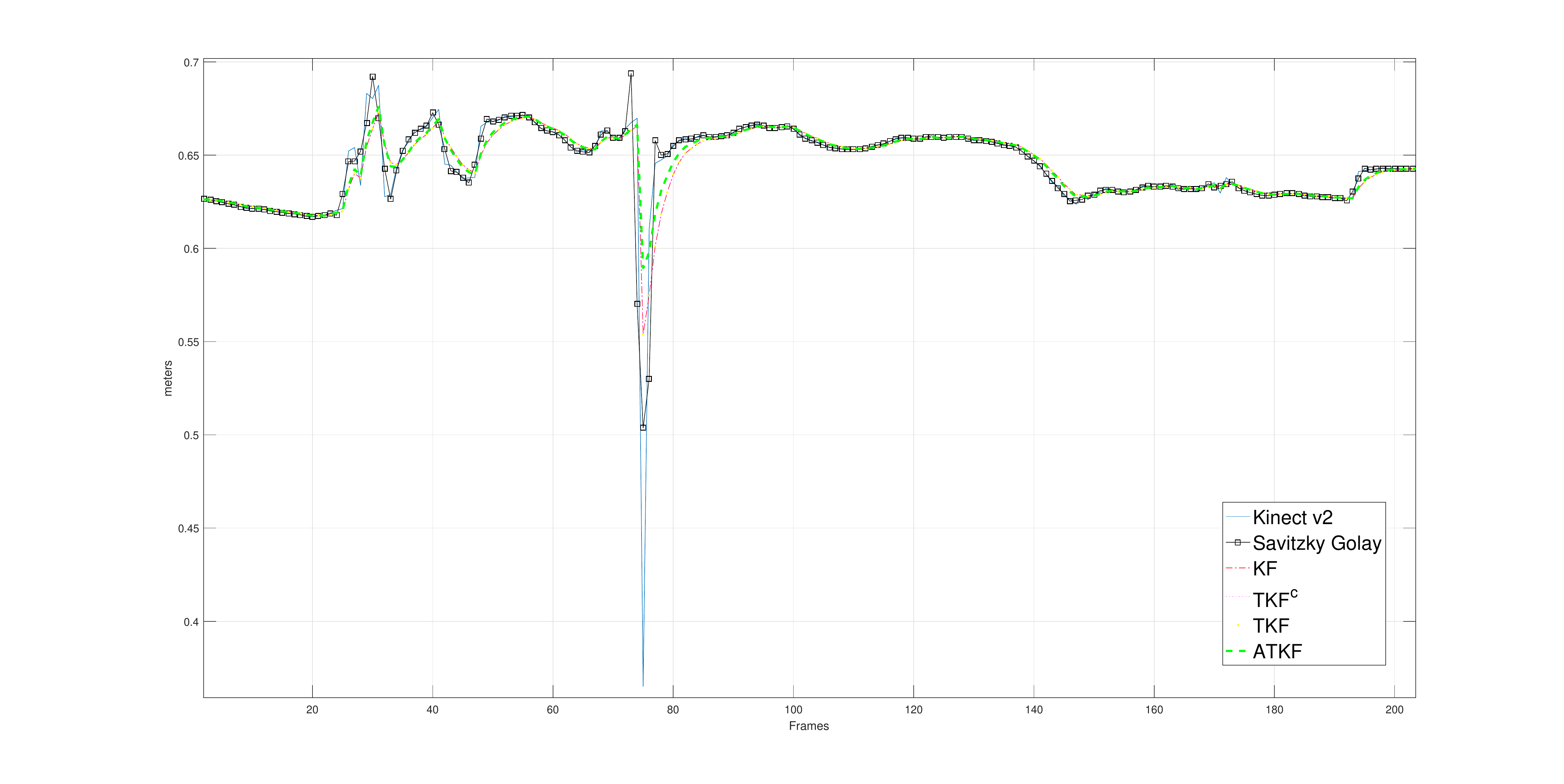}
    }
	\caption{The head's spatial coordinates $y_k$ of Kinect V2 sensor, Saviztky-Golay, KF, TKF, TKF$^c $ and ATKF.}
	\label{fig:cordy}
\end{figure}

\begin{figure}[ht]
	\centering
%	\subfloat[The $ x_k $ coordinates of the right hand.]
	{
		\includegraphics[width=4in]{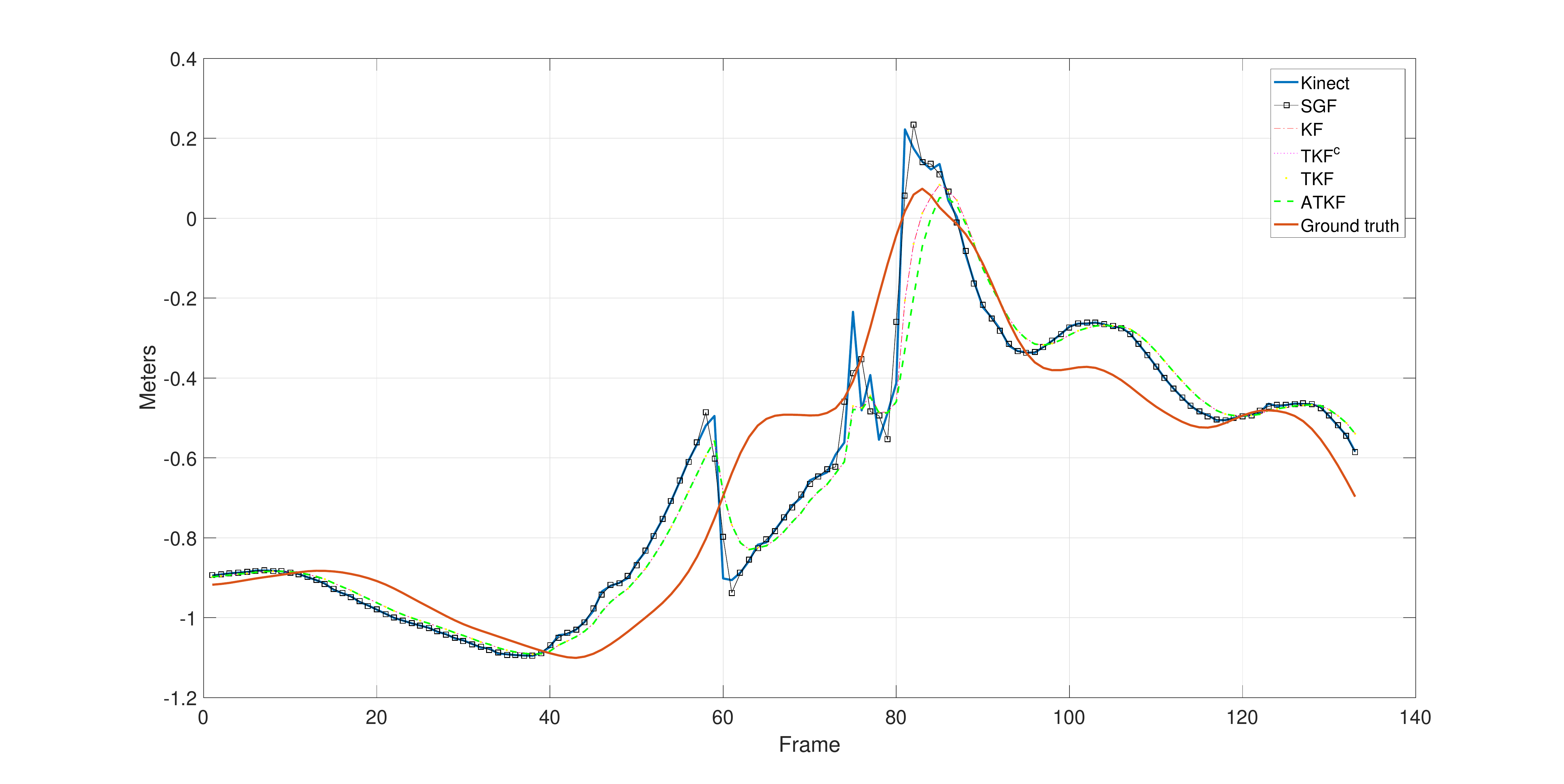} 
	}
    \hfill
%    \subfloat[The $ y_k $ coordinates of the right hand.]
    {
		\includegraphics[width=4in]{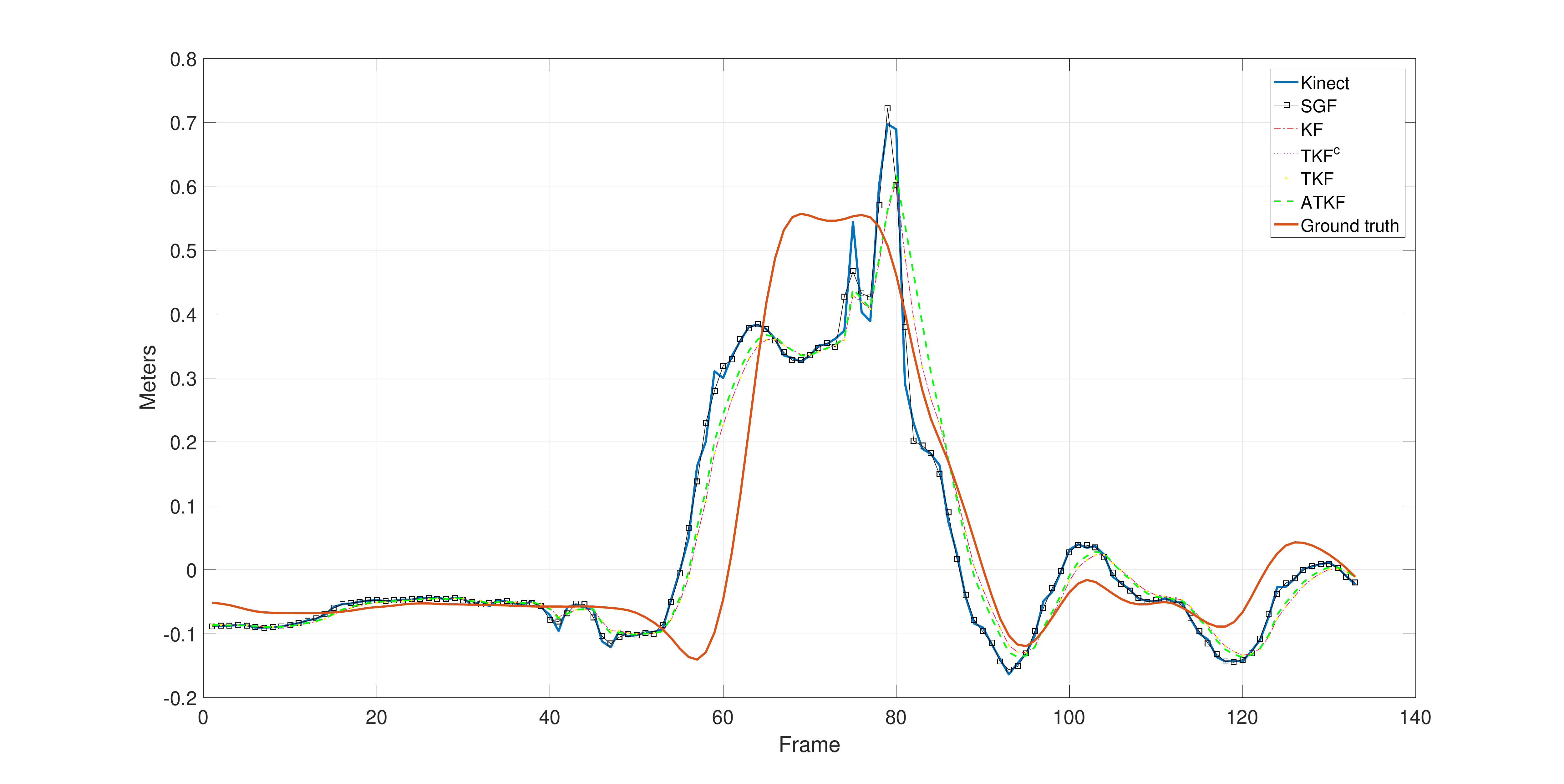} 
	}
	\hfill
%	\subfloat[The $ z_k $ coordinates of the right hand]
	{
		\includegraphics[width=4in]{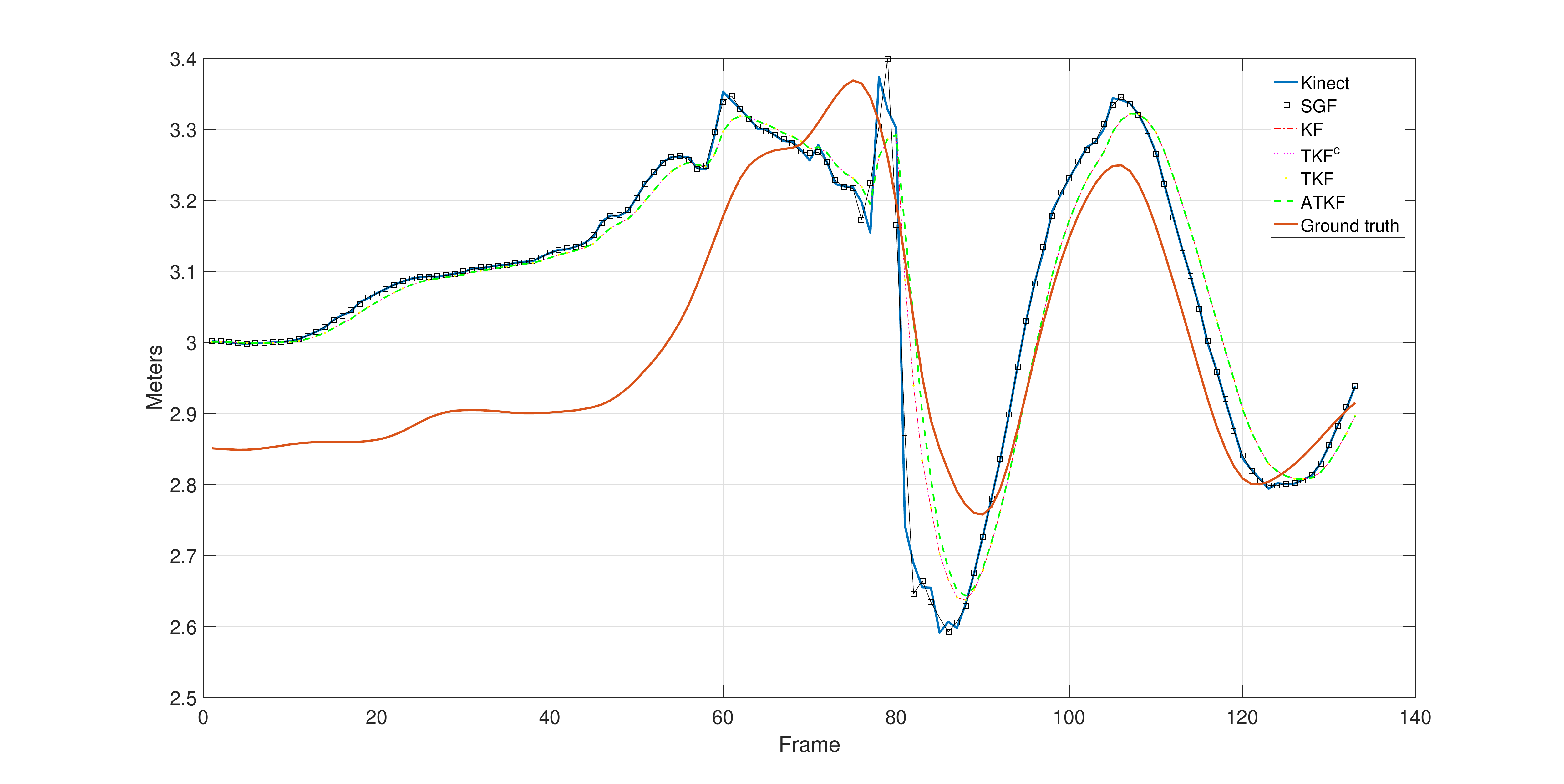} 
    }
	\caption{The right hand's coordinates by Kinect V2 sensor, SGF, KF, TKF$ ^c$, TKF, ATKF and Ground truth.}
	\label{fig:cordxyz}
\end{figure} 
To better illustrate the superiority of ATFK we illustrate the motion of the human skeleton (obtained by Kinect) under heavy occlusion in the first row of subfigures in Fig. \ref{fig:Atkf} for four consecutive frames. The first subfigure shows the human skeleton one frame before "collapsing", the next two show the human skeleton under heavy occlusion and the last one shows a better performance of  human skeleton. In the next five rows of Fig. \ref{fig:Atkf}, the motion of human skeleton is illustrated as it is resulted by SGF, KF, TKF, TKF$^c$ and ATKF, respectively. All filters had a delay of 1-2 frames due to the occluded area but ATKF clearly outperforms all other methods (see the last row in Fig. \ref{fig:Atkf})

\begin{table}[ht]
	\renewcommand{\arraystretch}{1.3}
	\begin{center}
		\begin{tabular}{ |c|c|c|c|c|c|c| }
			\hline
			\textbf{Angles}  & \textbf{Kin. v2} & \textbf{SGF} & \textbf{KF}& \textbf{TKF} & \textbf{TKF}$^c $& \textbf{ ATKF} \\
			\hline
			Right Elbow & 39.31 &37.44 & 36.60 & 36.60 & 36.60 & \textbf{36.32} \\
			Left Elbow & 31.58 & 30.65 &  27.98 & 27.98 & 27.98 & \textbf{26.50} \\
			Right Knee & 16.70 & 16.79 & 15.79 & 15.79 & 15.79 & \textbf{14.90}\\
			Left Knee & 26.25 & 25.81 & 25.14 &  25.14 & 25.14 & \textbf{25.11}\\
			\hline
		\end{tabular}
	\end{center}	
	\begin{center}
		\begin{tabular}{ |c|c|c|c|c|c|c| }
			\hline
			\textbf{Angles}  & \textbf{Kin. v2} & \textbf{SGF} & \textbf{KF}& \textbf{TKF} & \textbf{TKF}$^c $& \textbf{ ATKF} \\
			\hline
			Right Elbow & 38.76 & 36.86 & 35.90 & 35.90 & 35.90 & \textbf{35.57} \\
			Left Elbow & 32.18 & 31.27 & 28.43 & 28.43 & 28.43 &  \textbf{27.02}\\
			Right Knee & 17.03 & 17.12 & 15.75 & 15.75 & 15.75 &  \textbf{14.93}\\
			Left Knee & 26.38 &26.01 & 24.85 & 24.85 & 24.85 &  \textbf{24.82}\\
			\hline
		\end{tabular}
	\end{center}	
	\begin{center}
		\begin{tabular}{ |c|c|c|c|c|c|c| }
			\hline
			\textbf{Angles}  & \textbf{Kin. v2} & \textbf{SGF} & \textbf{KF}& \textbf{TKF} & \textbf{TKF}$^c $& \textbf{ ATKF} \\
			\hline
			Right Elbow & 38.43 & 36.63 & 35.40 & 35.40 & 35.40 & \textbf{35.06}\\
			Left Elbow & 32.99 & 32.09& 29.08 & 29.08 & 29.08 & \textbf{27.75}\\
			Right Knee & 17.77 & 17.79 & 16.04 & 16.04 & 16.04 & \textbf{15.26}\\
			Left Knee & 26.67 & 26.46 &  24.90 & 24.90 & 24.90 & \textbf{24.89}\\
			\hline
		\end{tabular}
	\end{center}		
	\begin{center}
		\begin{tabular}{ |c|c|c|c|c|c|c| }
			\hline
			\textbf{Angles}  & \textbf{Kin. v2} & \textbf{SGF} & \textbf{KF}& \textbf{TKF}& \textbf{TKF}$^c $ & \textbf{ ATKF} \\
			\hline
			Right Elbow & 38.39 &36.64& 35.25 & 35.25 &  35.25 & \textbf{ 34.93}\\
			Left Elbow & 33.96 & 33.06  & 29.92 & 29.92 & 29.92 & \textbf{28.70}\\
			Right Knee & 18.78& 18.78 & 16.58 & 16.58 & 16.58 & \textbf{15.77}\\
			Left Knee & 27.14& 27.02 & 25.24 & 25.24 & 25.24 & \textbf{25.23}\\
			\hline
		\end{tabular}\\
	\end{center}
	\caption{RMSEs for the angles by Kinect V2, SGF, KF, TKF, TKF$^c$ and ATKF for time delay 92, 93, 94 and 95.}
	\label{tab:RMSE}
\end{table}

\subsection{Recording by Kinect Sensor and Vicon System}
In this subsection, we evaluate the proposed method with respect to ground truth data. To that end, we monitor an athlete throwing a ball with his right hand, and we record this motion by a Kinect V2 sensor and the Vicon system at the same time. We use Vicon as the ground truth in order to compare results using  the proposed method on Kinect measurements. The number of Kinect's and Vicon's frames are 266 (almost 8.8667 sec.) and 139 (4.4480 sec.), respectively. We note that Kinect time-stamp  is almost 0.033 sec per frame while Vicon time-stamp is constantly 0.032sec. We interpolate  Vicon data in order to deal with the time-stamp problem; after interpolation, the new Vicon data include 133 frames. Therefore, we temporally synchronize the two sensors to start  together. To do so, we initially calculate the angles of knees and elbows obtained by Kinect and Vicon data and then, we calculate the RMSE between these angles for different delays. The results show that the minimum values of RMSE for every angle appeared for delays of 92-95 frames. The different delays between the angles in some cases are somewhat expected because Kinect records fast movements with delay (i.e., after some frames).

We notice that KF smooths the spatial coordinates without affecting the movement (see Fig. \ref{fig:cordxyz}). TKF$^{c}$ and TKF perform exactly the same smoothing in all joints as KF, while SGF does not perform a satisfactory smoothing in some points where the measurements have a significant error. In  Table \ref{tab:RMSE} we observe the RMSEs for the angles as they arise for delays $ t=92,93,94,95 $ frames, respectively. In all cases, the RMSEs are big enough because of the occlusion of some joints during the recording.

In  Fig.\ref{fig:cordxyz}  the right hand's coordinates resulted by KF, TKF, TKF$^c $ and ATKF are almost the same, because all measurements belong to the uncensored region, while SGF coordinates are almost the same with Kinect's coordinates. However, as can be seen in  Table \ref{tab:RMSE} , in all cases concerning RMSEs, we get better results via ATKF compared to those of standard KF,  TKF$^c $ and TKF. The RMSEs of SGF  are almost the same as the Kinect RMSEs.

\section{Conclusion and Discussion}

The aim of this paper was to improve 1) the well-known TKF process \cite{allik2014tobit} and 2) the human skeleton motion tracking using a single Kinect V2 sensor, which often generates noisy measurements due to occlusion, lighting conditions, etc. To that end, we proposed a novel filtering method, called ATKF, which relies on the censored data statistics theory for human skeleton motion tracking in real-time. In order to estimate the hidden state vector by the censored measurement, firstly, we evaluated the probabilities of a latent  measurement to belong in or out of the uncensored region (Appendix C) and secondly, we evaluated the accurate covariance matrix of  the censored normal distribution (Appendix B). In this approach, we had to define the limits of the uncensored region for the Kinect's measurements, in a reasonable manner for every time step $ k $. To do so, we tested many data with various joints movements, which were  obtained by ground truth sensor, such as the Vicon tracking system. 

We evaluated the proposed method against 1) standard KF, 2) TKF, 3) TKF$^c$ with constant limits and 4) SGF in three different setups: 1) Artificial data 2) Kinect and 3) Kinect plus Vicon human skeleton motion data. We also introduced a new metric in order to evaluate results when no ground truth is available. Finally, we calculated the covariance matrix, $ \textbf{Q}$, of the noise process under a specific experimental methodology as opposed to previous methods where random or simple experimental covariance matrices were used. Among the five approaches, ATKF gave better results in all the different setups for human skeleton tracking. 

In a future work it would be interesting to use the proposed filtering method for action recognition tasks in the wild, where uncontrolled environments and situations where RGB-D sensors may have poor performance often occur. Moreover, as a step beyond, it would be interesting to consider the state vector $\mathbf{x}$ as a censored state, aiming at achieving a more accurate filtering of the human skeleton motion data. 

\clearpage 
\begin{figure*}[h!] 
	\centering
	\begin{tabular}{ |l| c c c c| }	
		\hline	
		\textbf{1. Kinect} & \includegraphics[width=1.2in]{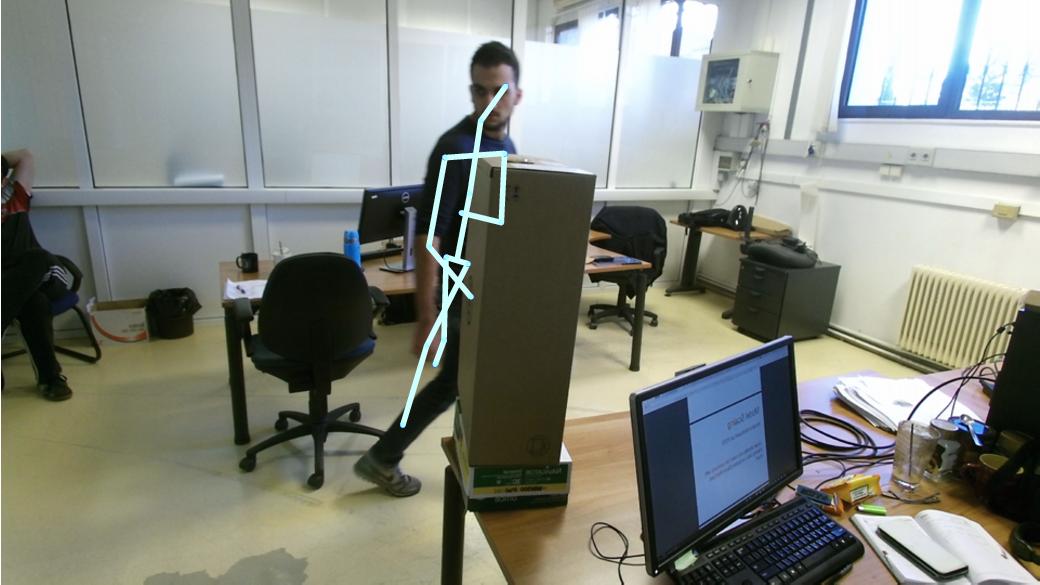}& \includegraphics[width=1.2in]{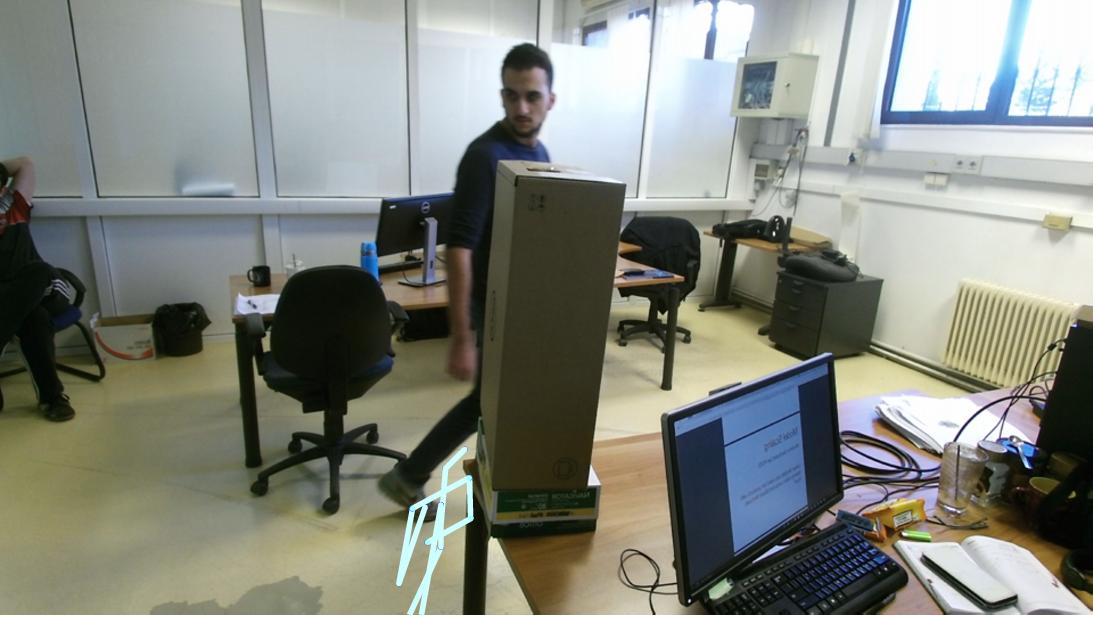} &
		\includegraphics[width=1.2in]{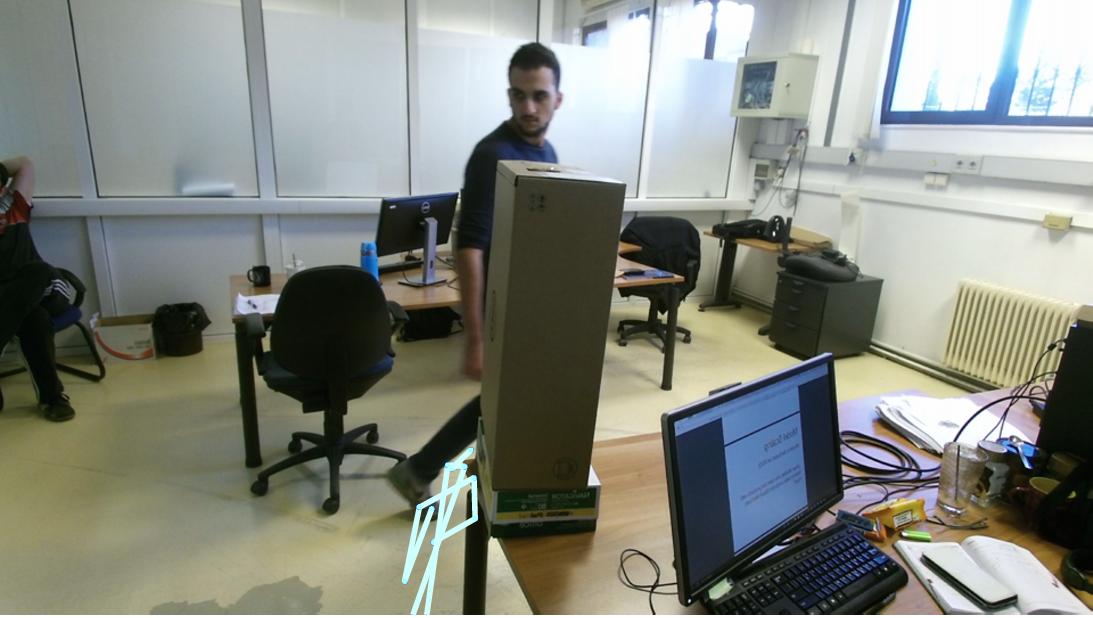} & \includegraphics[width=1.2in]{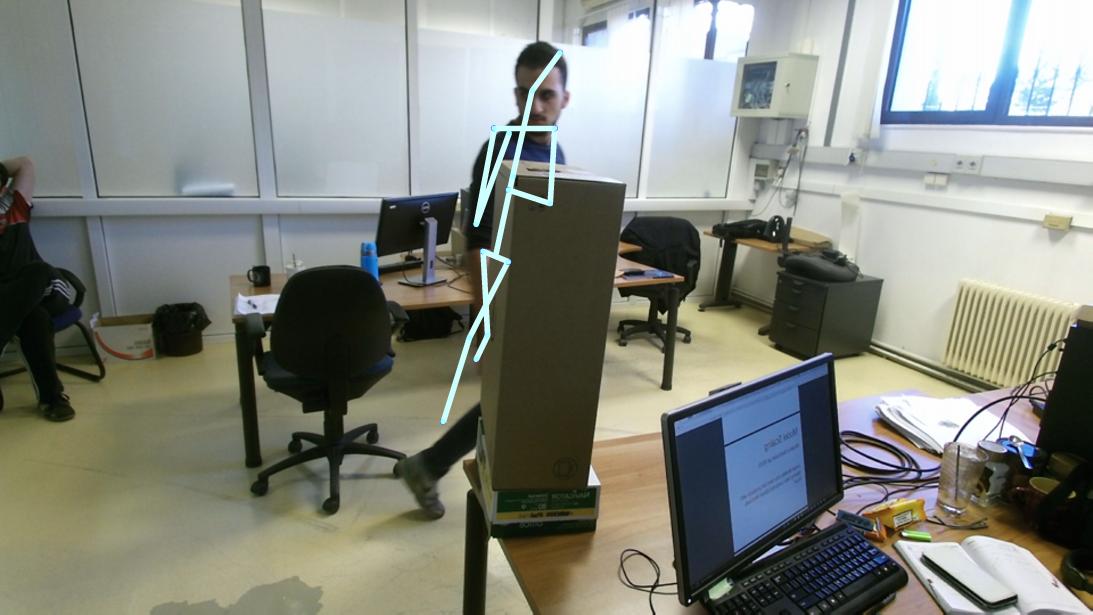} \\
		\hline	
		\textbf{2. SGF} & \includegraphics[width=1.2in]{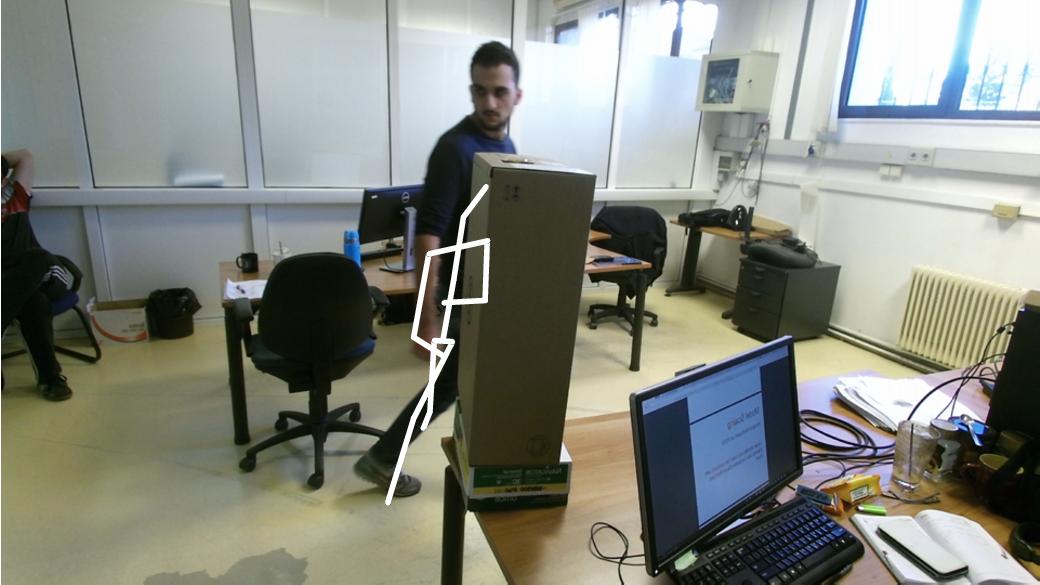}& \includegraphics[width=1.2in]{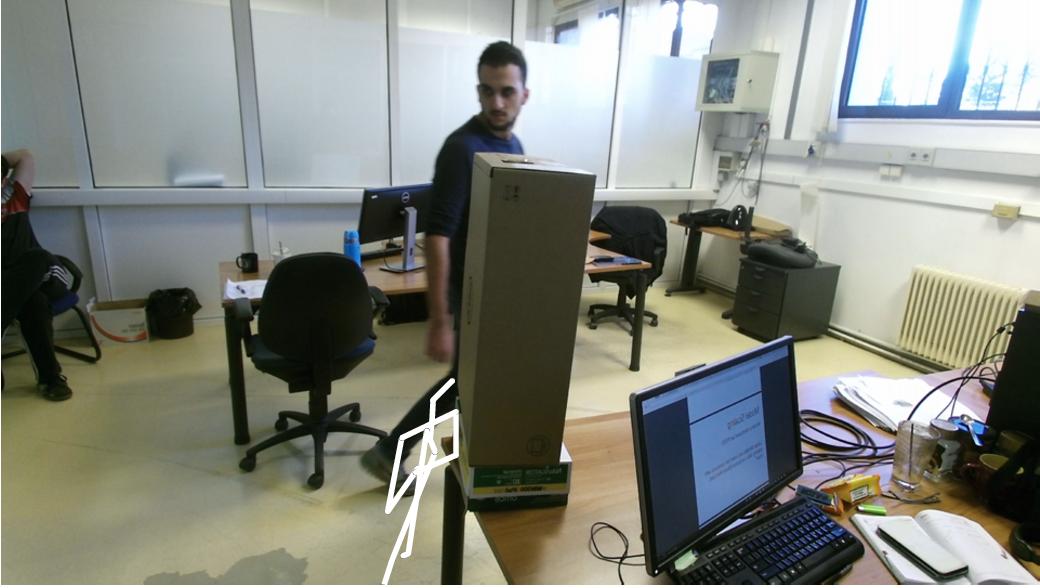} & \includegraphics[width=1.2in]{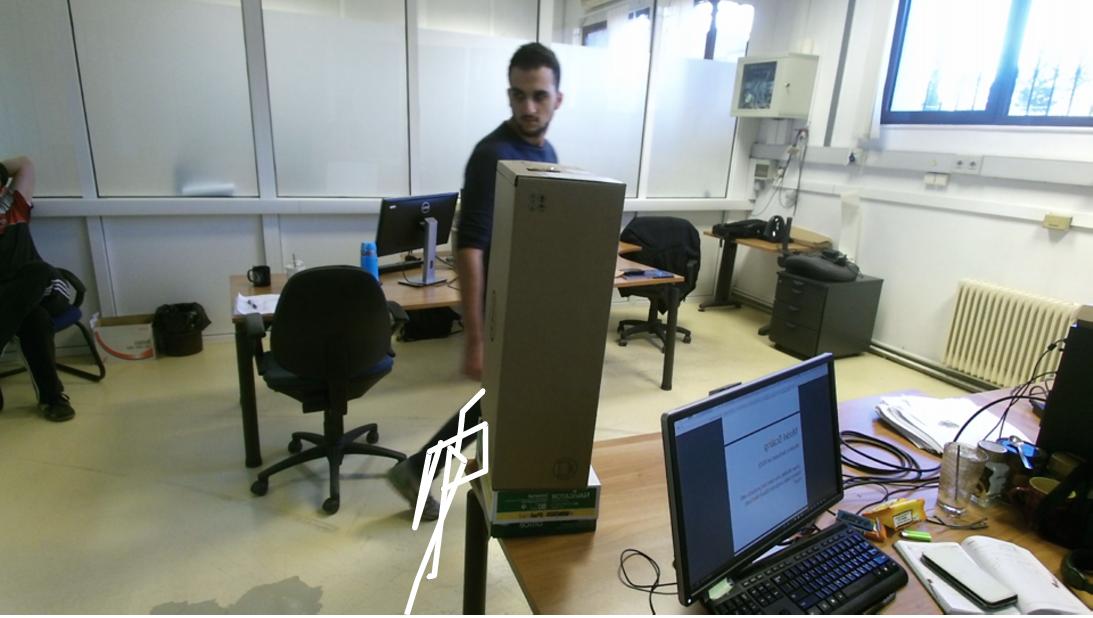} & \includegraphics[width=1.2in]{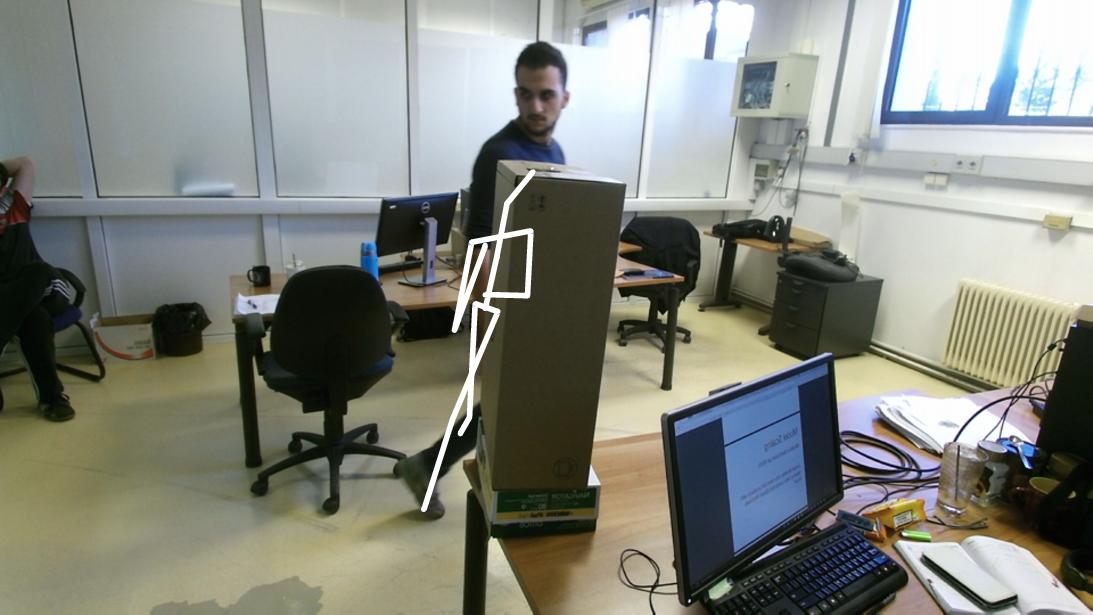} \\
		\hline	
		\textbf{3. KF} & \includegraphics[width=1.2in]{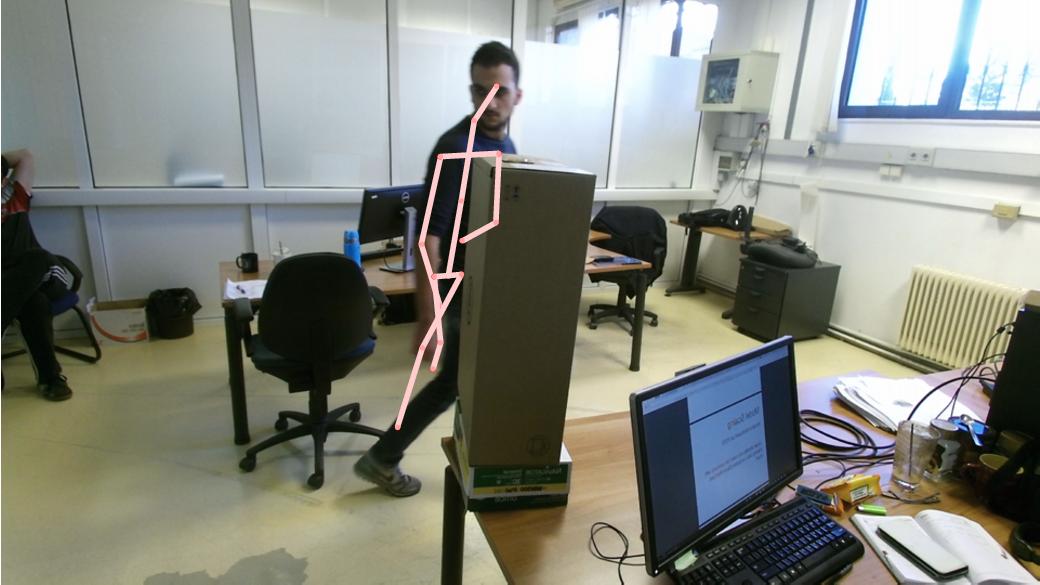}& \includegraphics[width=1.2in]{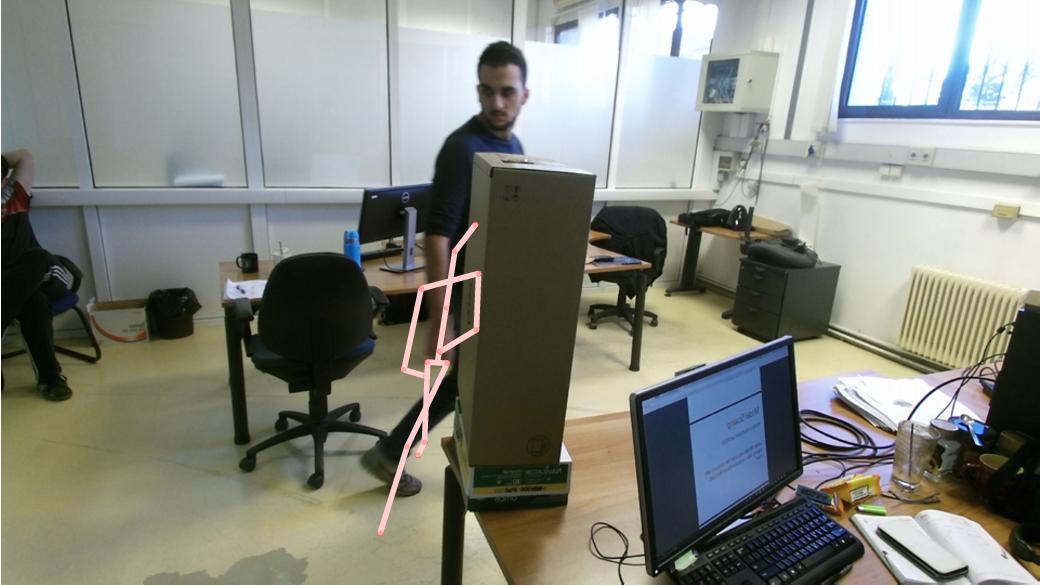} & \includegraphics[width=1.2in]{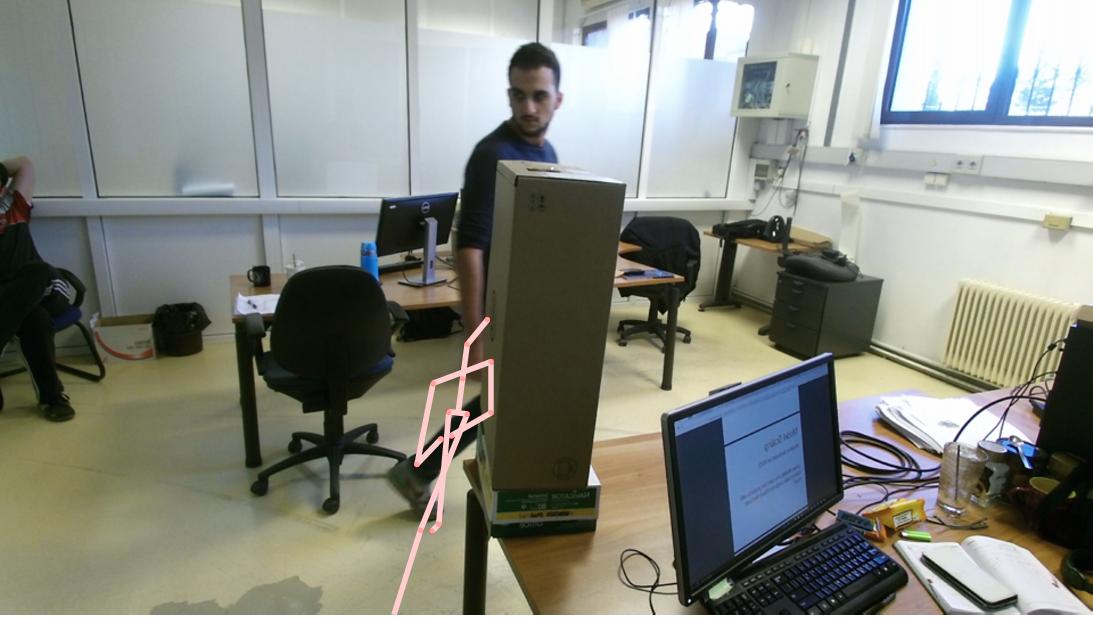} & \includegraphics[width=1.2in]{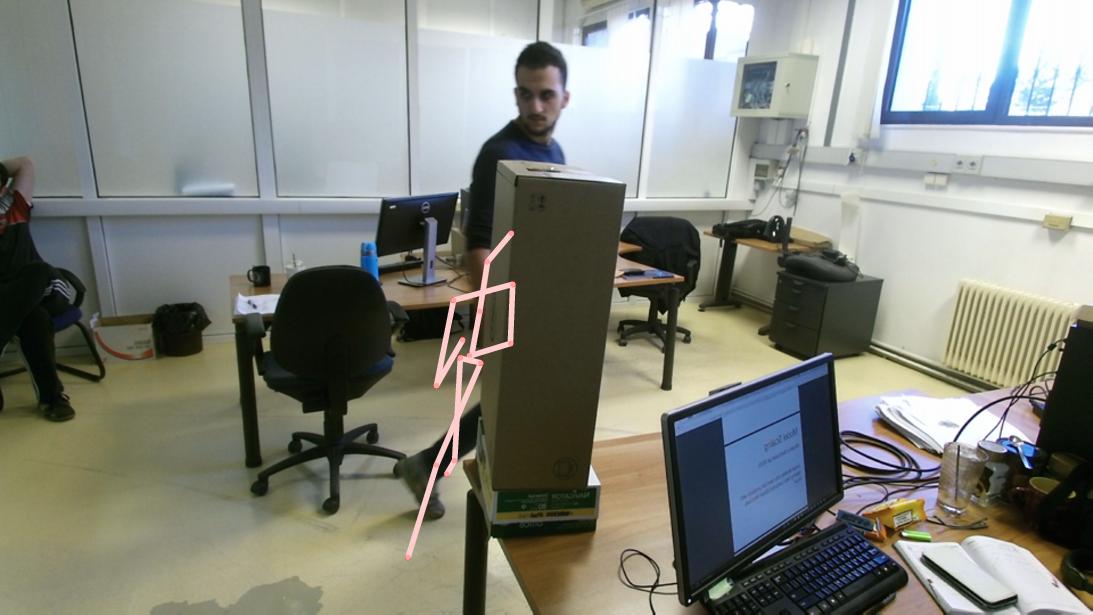} \\
		\hline	
		\textbf{4. TKF} & \includegraphics[width=1.2in]{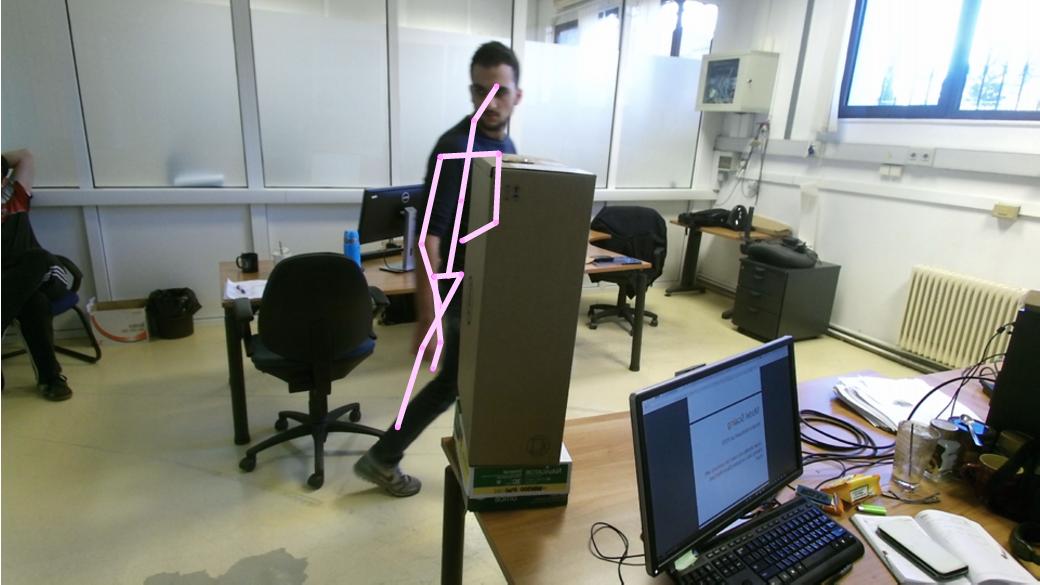}& \includegraphics[width=1.2in]{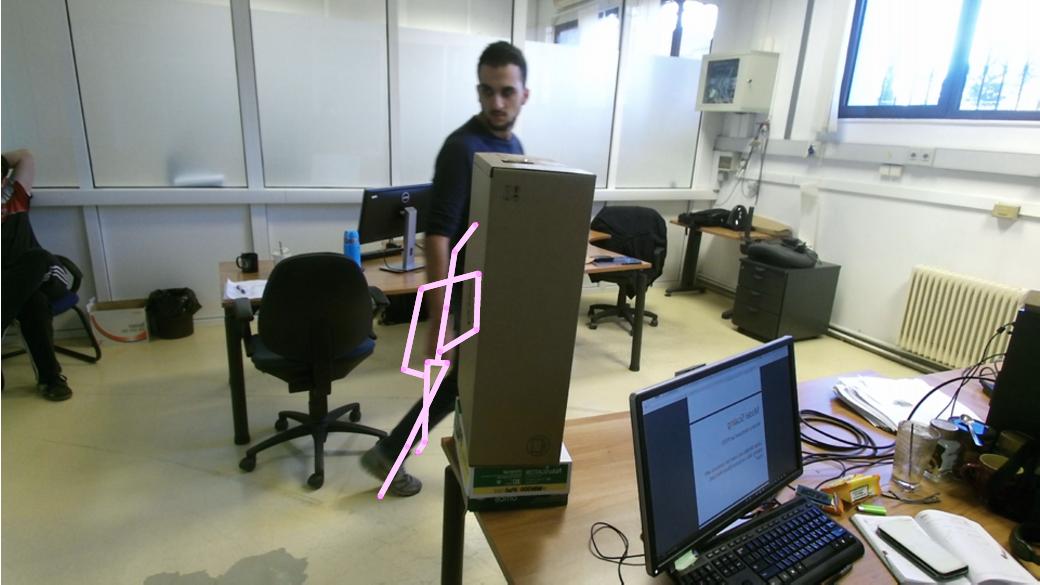} & \includegraphics[width=1.2in]{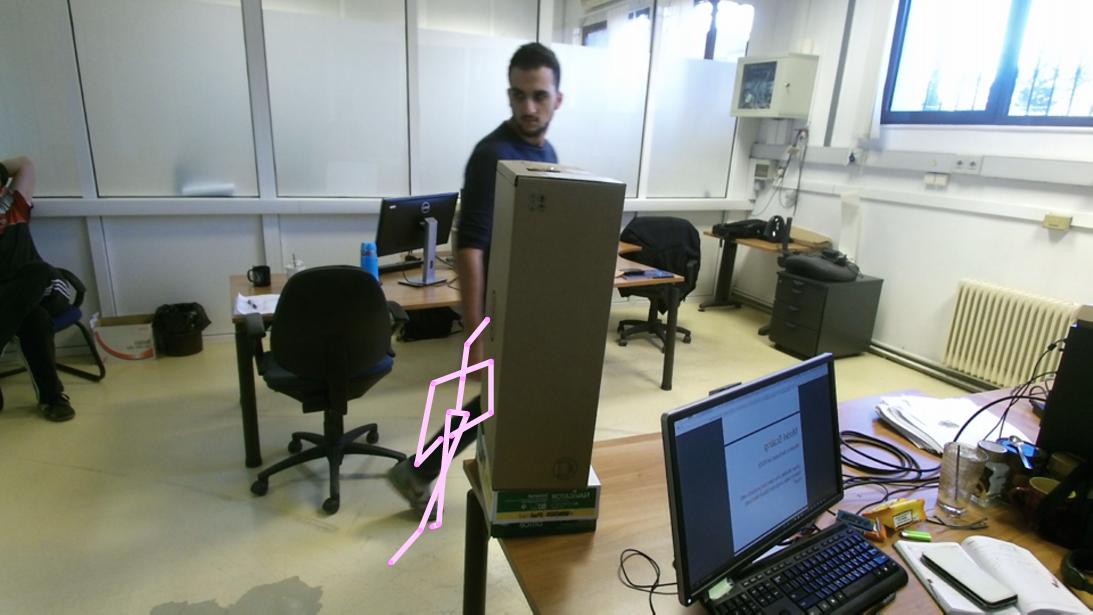} & \includegraphics[width=1.2in]{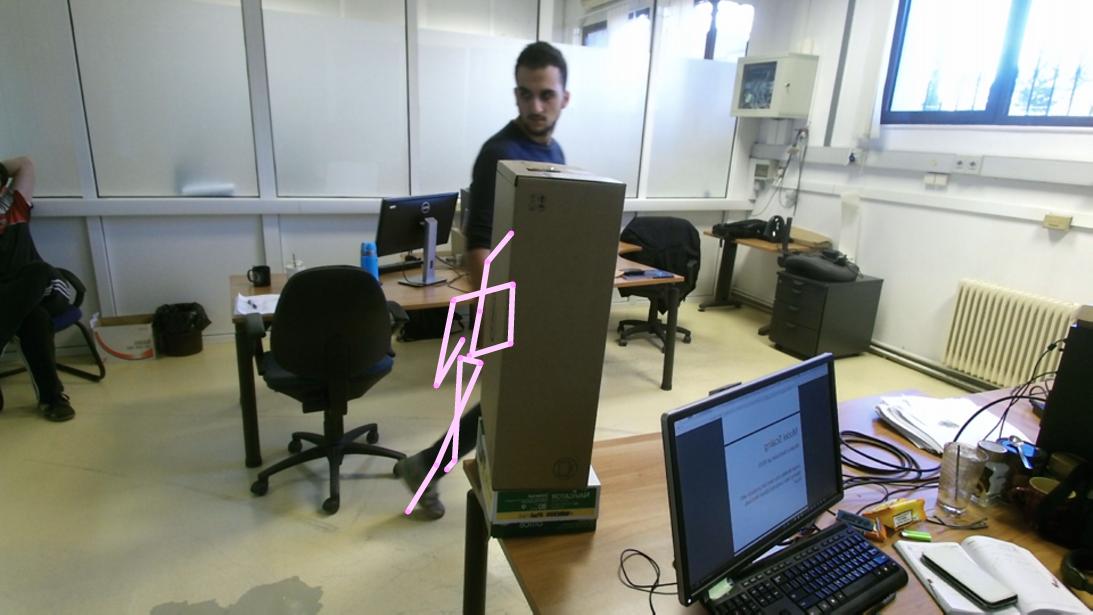} \\
		\hline	
		\textbf{5. TKF}$^c$ & \includegraphics[width=1.2in]{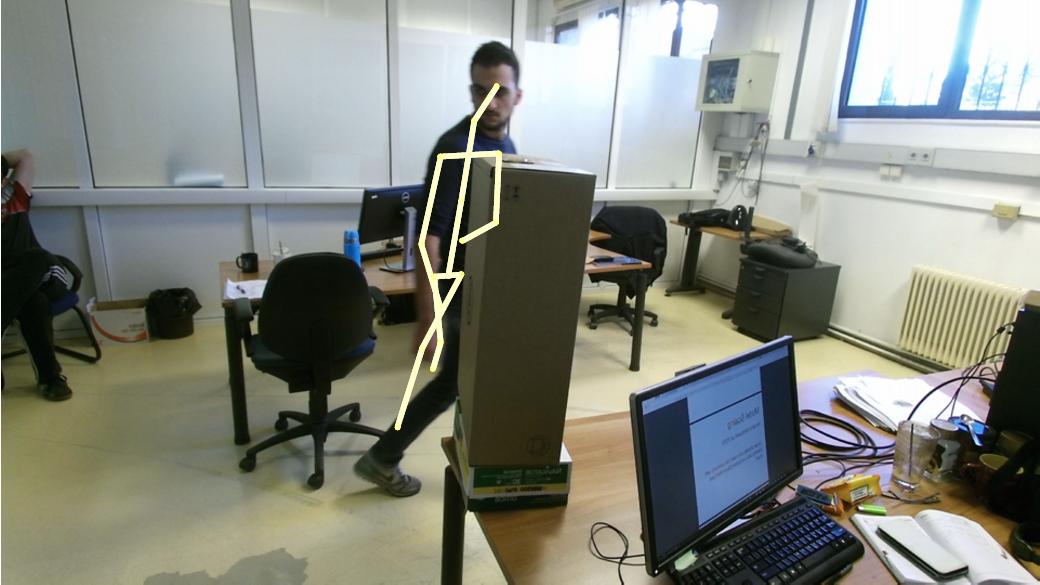}& \includegraphics[width=1.2in]{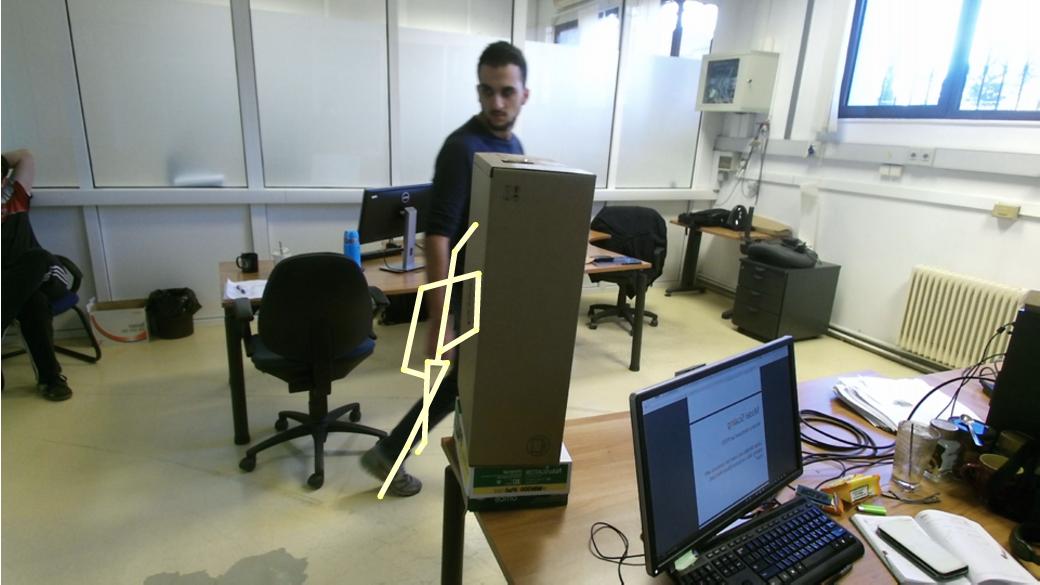} & \includegraphics[width=1.2in]{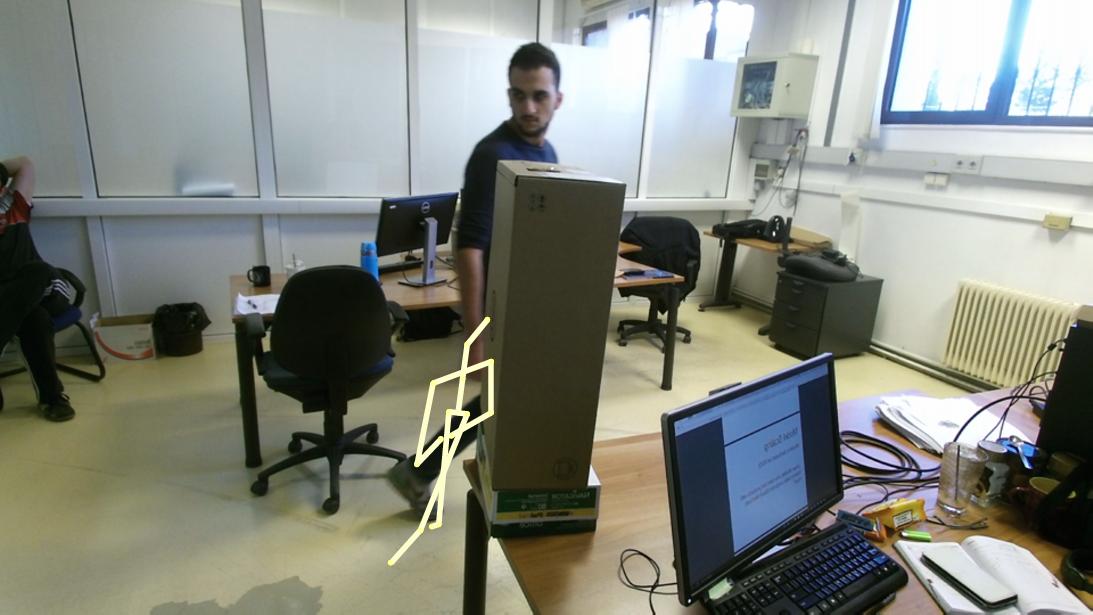} & \includegraphics[width=1.2in]{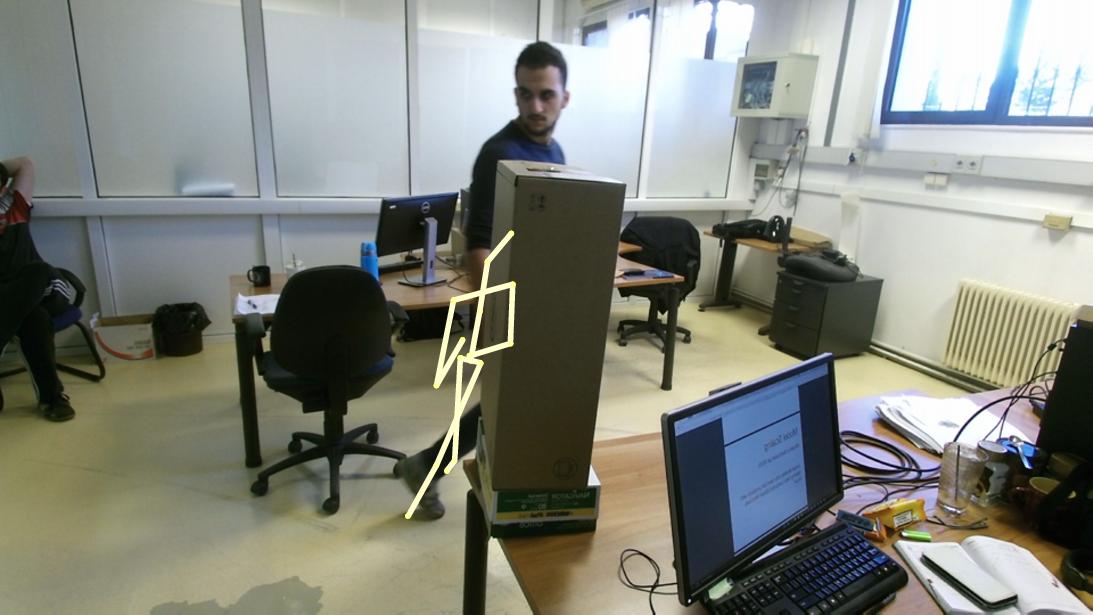} \\
		\hline	
		\textbf{6. ATKF} & \includegraphics[width=1.2in]{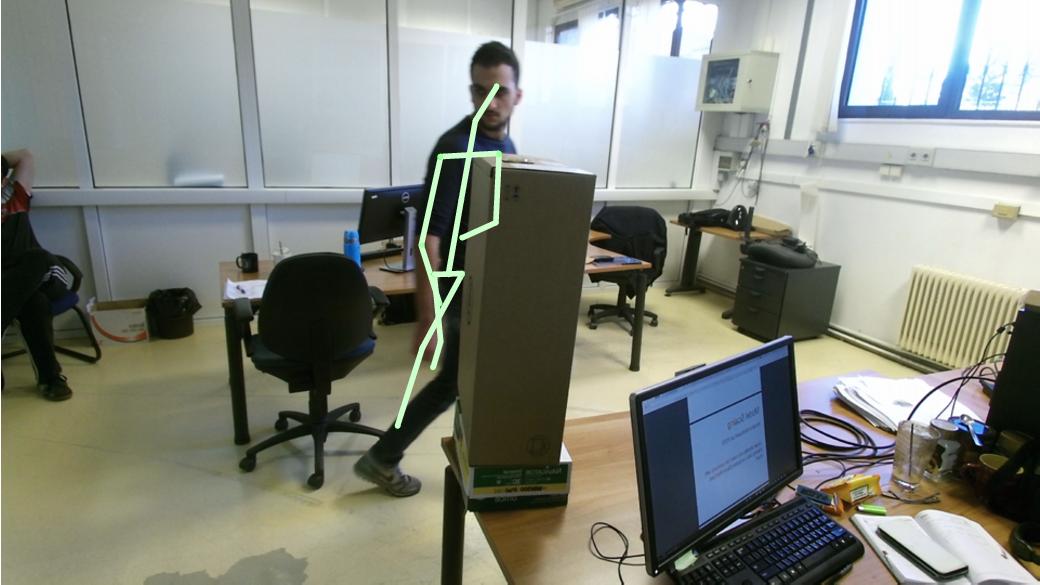}& \includegraphics[width=1.2in]{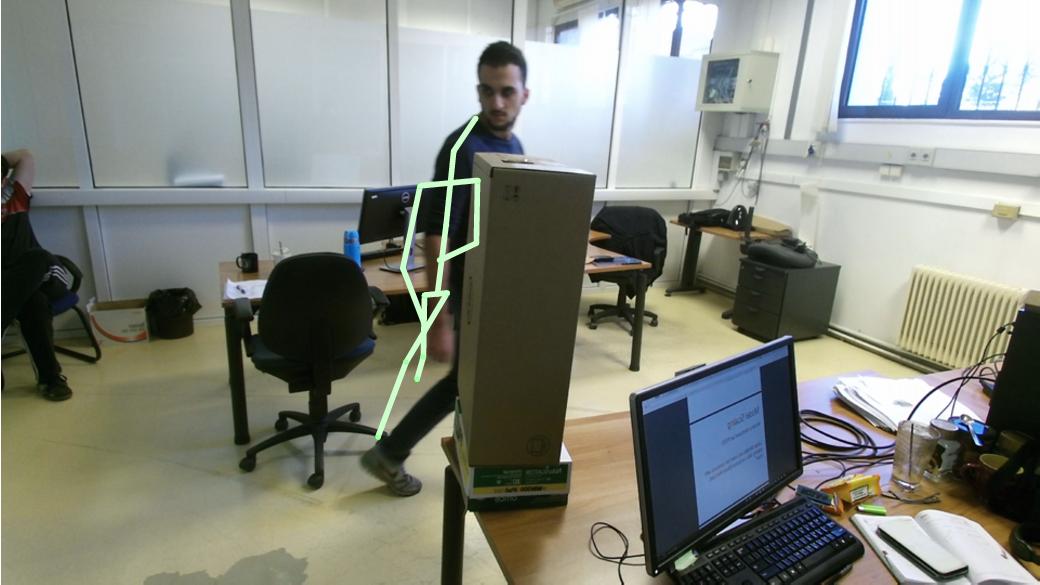} & \includegraphics[width=1.2in]{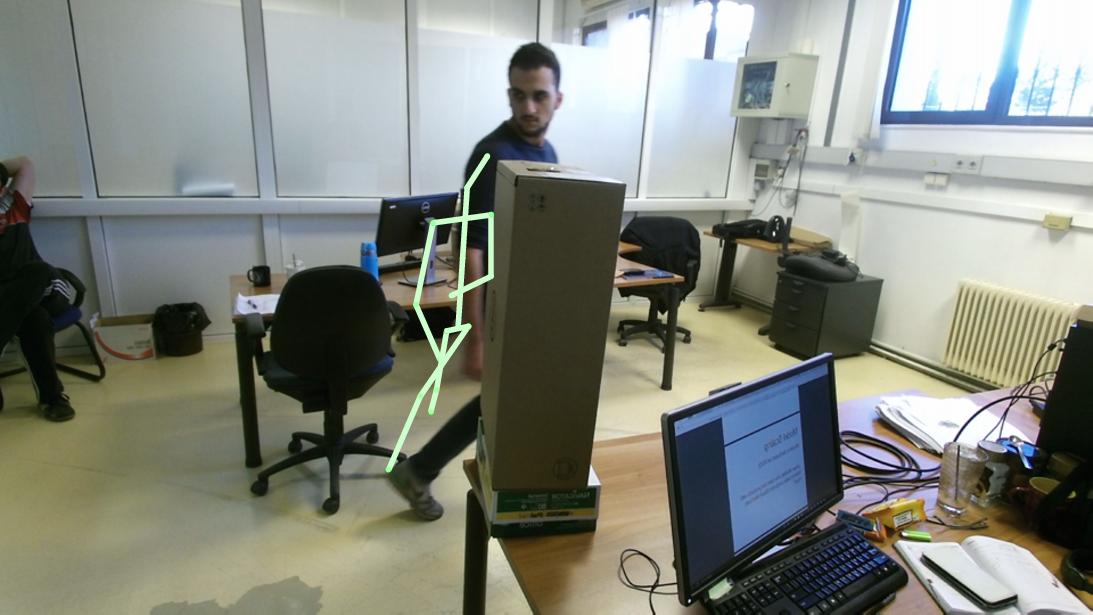} & \includegraphics[width=1.2in]{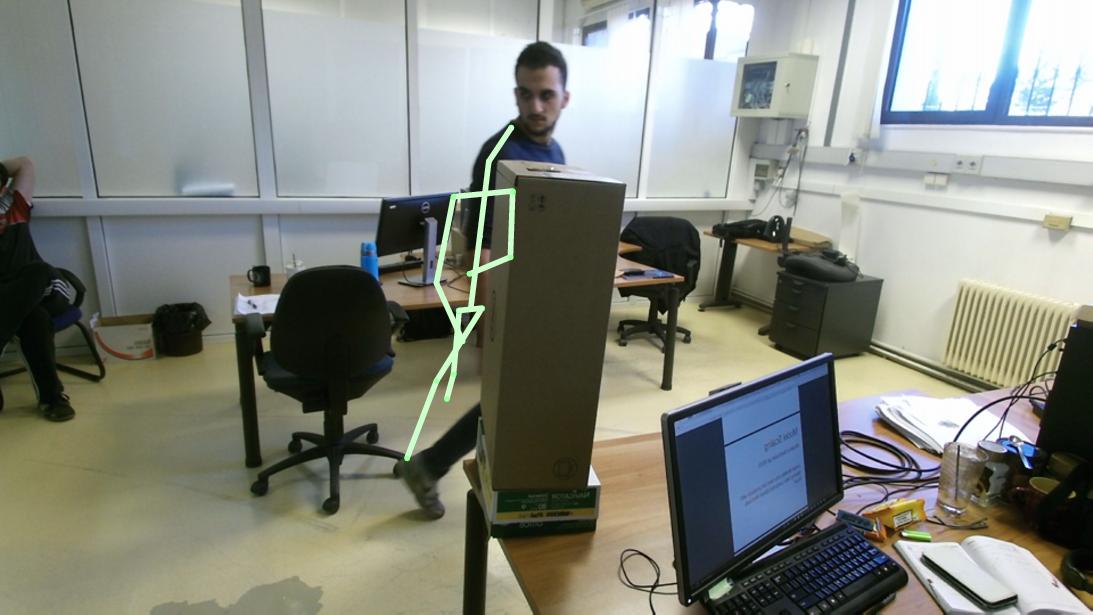} \\
		\hline			
	\end{tabular}
	\caption{Each row represents the human skeleton motion for four consecutive frames as it is obtained by 1) Kinect V2 sensor, 2) SGF, 3) KF, 4) TKF, 5) TKF$^c$ and 6) ATKF, respectively.}  
	\label{fig:Atkf}
\end{figure*}

\section*{Appendix A: The censored mean value} 

\label{App_A}
For a discrete random variable $ Z_i \sim B(p_i) $ (Bernoulli distribution) in Lemma \ref{joint_mean}, it is derived that 
\begin{equation}
E(X,Z_i) = E(X|Z_i=1)\cdot p_i.
\label{app_lemma}
\end{equation}
The censored measurement, $ y_i $ can be written in terms of Bernoulli distributions, therefore, the censored mean value is written by (\ref{app_lemma}) as,
\begin{equation}
\begin{split}
E(y_i) = \sum_{j=1}^{3^{n-1}}\E(y^*_i|a_i<y^*_i<b_i, R_j)P((a_i,b_i),R_j) + a_iP(y^*_i\leq a_i) +b_iP(y^*_i \geq b_i),
\end{split}
\label{ap:cen_mean}
\end{equation}
where the first term is the sum of all possible mean values of $ E(y_i|a_i<y_i<b_i)$ given that the rest variables lie in a region\\ $ R_j = [(L_1,U_1),...,(L_{i-1},U_{i-1}),(L_{i+1},U_{i+1}),...,(L_n,U_n)] $, where 
\[ 
(L_k,U_k)=
\begin{cases}
(-\infty,a_k) \quad  or\\
(a_k, b_k)    \quad \; \; \; or\\
(b_k, \infty)
\end{cases} 
\]
where j=1,...,$ 3^{n-1} $. $P((a_i,b_i),R_j) $ denotes the probability of variable $ \textbf{y}^* $ to lie in a region  $[(a_i,b_i),R_j]. $
It is derived by (\ref{trun_mean}) that 

\begin{equation}
\begin{split}
E(y_i) =& \sum_{j=1}^{3^{n-1}}\Big(\mu_i + \sum_{k=1}^n \sigma_{i,k}\big(F_{k}(L_k)-F_{k}(U_k)\big)_{R_j}  \Big)P(R_j)+ a_iP(y^*_i\leq a_i) +b_iP(y^*_i \geq b_i)\\
=& \sum_{k=1}^n\sum_{j=1}^{3^{n-1}}\sigma_{i,k}\big(F_{k}(L_k)-F_{k}(U_k)\big)_{R_j}P(R_j)+\mu_iP(a_i<y^*_i<b_i) + a_iP(y^*_i\leq a_i)\\ &+b_iP(y^*_i \geq b_i),
\end{split}
\label{ce_mean2}
\end{equation}
where $ \big(F_{k}(L_k)-F_{k}(U_k)\big)_{R_j} $ is the difference of functions (\ref{F}) given that the variable $ \textbf{y}^* $ lies in the region $ R_j \cup (a_i,b_i) $.\\ 
In the case where $ k \neq i $, it is derived that: 

\begin{equation}
\begin{split}
\sum_{j=1}^{3^{n-1}}\sigma_{i,k}\big(F_{k}(L_k)-F_{k}(U_k)\big)_{R_j}P(R_j)&=\sum_{j=1}^{3^{n-2}}\sigma_{i,k}\big(F_{k}(-\infty)-F_{k}(a_k)\big)_{V_j}P(V_j)\\
&+\sum_{j=1}^{3^{n-2}}\sigma_{i,k}\big(F_{k}(a_k)-F_{k}(b_k)\big)_{V_j}P(V_j)\\
&+\sum_{j=1}^{3^{n-2}}\sigma_{i,k}\big(F_{k}(b_k)-F_{k}(\infty)\big)_{V_j}P(V_j)=0 
\end{split}
\label{first}
\end{equation}
where $V_j $ is the region 
\[ 
[(L_1,U_1),...,(L_{k-1},U_{k-1}),(L_{k+1},U_{k+1}),...,(a_i,b_i),...,(L_n,U_n)]. 
\]
In the case where $ k = i $, it is derived that, 
\[
\sum_{j=1}^{N}\sigma_{i,k}\big(F_{k}(L_k)-F_{k}(U_k)\big)_{R_j}P(R_j) =  
\]

\begin{equation}
\begin{split}
&=\sum_{j=1}^{N}\sigma_{i,i}\big(F_{i}(a_i)-F_{i}(b_i)\big)_{R_j}P(R_j)\\
&=\sigma_{i,i}\big(f_{i}(a_i)-f_{i}(b_i)\big)
\end{split}
\label{second}
\end{equation}
where $ f_i(y^*_i) $ is the normal distribution of $ y^*_i \sim N(\mu_i, \sigma_{i,i}) $. Thus, by (\ref{ce_mean2})-(\ref{second}) , we have

\begin{equation}
\begin{split}
E(y_i) = \mu_iP(a_i<y^*_i<b_i)  + \sigma_{i,i}(f_i(a_i)-f_i(b_i)) +a_iP(y^*_i\leq a_i) + b_iP(y^*_i \geq b_i). 
\end{split}
\label{app_mean}
\end{equation}

\section*{Appendix B: The censored covariance matrix} 
\label{App_B}

In the same way as in censored mean (Appendix A), it is proved that the second moment of $ y_i $ is dependent only on the censoring limit \{$ a_i $, $ b_i $\}. Therefore, it is derived by Lemma \ref{joint_mean} that

\[
\begin{split}
E\big(y_i^2\big) = E(y^{*2}_i|a_i<y^*_i<b_i)P(a_i<y^*_i<b_i) + a^2_iP(y^*_i \leq a_i) + b^2_iP(y^*_i \geq b_i)
\end{split} 
\] 
where the first term \cite{bg2009moments} is equal with

\begin{equation}
\begin{split}
E(y^{*2}_i|a_i<y^*_i<b_i) &= \sigma_{i,i} + \mu_i^2 + 2\mu_i\sigma_{i,i}\frac{f_i(a_i)-f(b_i) }{P(a_i<y^*_i<b_i)} \\ &+\sigma_{i,i}\frac{(a_i-\mu_i)f_i(a_i)- (b_i-\mu_i)f_i(b_i) }{P(a_i<y^*_i<b_i)} 
\label{moment2}
\end{split}
\end{equation}
Therefore, it is derived by (\ref{moment2}) that,

\begin{equation}
\begin{split}
E\big(y_i^2\big) &= (\sigma_{i,i} + \mu_i^2)P(a_i<y^*_i<b_i) \\
&+ \sigma_{i,i}\big((a_i-\mu_i)f_i(a_i)-  (b_i-\mu_i)f_i(b_i)\big)\\
& + 2\mu_i\sigma_{i,i}(f_i(a_i)-f(b_i)) + a^2_iP(y^*_i \leq a_i) + b^2_iP(y^*_i \geq b_i) 
\end{split}
\end{equation}
Finally the censored variance is given by

\begin{equation}
\begin{split}
Var(y_i) &= \mu_i^2(1-P_{un,i})P_{un,i} + \sigma_{i,i}P_{un,i} + a^2_i(1-P_{a,i})P_{a,i}  \\
&+ b^2_i(1-P_{b,i})P_{b,i}-2a_ib_iP_{a,i}P_{b,i} - \sigma_{i,i}^2(f(a_i)-f(b_i))\\ 
& +2\mu_i\sigma_{i,i}(f_i(a_i)-f(b_i))(1-P_{un,i})\\
&+\sigma_{i,i}\big((a_i-\mu_i)f_i(a_i)- (b_i-\mu_i)f_i(b_i)\big) \\
&-2\Big( \mu_iP_{un,i} + \sigma_{i,i}\big(f_i(a_i)-f(b_i)\big)\Big)\Big(a_iP_{a,i} + b_iP_{b,i} \Big)\\ 
\end{split}
\label{app_var}
\end{equation}

where $ P_{un} = P(a_i<y^*_i<b_i) $, $ P_a= P(y^*_i \leq a_i)$ and $ P_b = P(y^*_i \geq b_i) $.

The expectation value of $ y_i \cdot y_j  $ is written by Lemma \ref{joint_mean} as:
\begin{equation}
\begin{split}
\E(y_{i}y_{j}) &=a_ib_j P(1) +b_ib_j P(3) + a_ia_jP(7) + b_ia_jP(9)\\
& + b_j\sum_{k=1}^{3^{n-2}}\E(y^*_{i}|a_i< y^*_{i} <b_i, y^*_j \geq b_j, G_{k})P(G_{k})\\
& + a_i\sum_{k=1}^{3^{n-2}}\E(y^*_{j}|a_j< y^*_{j} <b_j, y^*_i \leq a_i, G_{k})P(G_{k})\\
& +\sum_{k=1}^{3^{n-2}}\E(y_{i}y^*_{j}|a_i< y^*_{i} < b_i, a_j < y^*_{j} <b_j, G_{k})P(G_{k})\\
& + b_i\sum_{k=1}^{3^{n-2}}\E(y^*_{j}|a_j< y^*_{j} <b_j, y^*_i \geq b_i, G_{k})P(G_{k})\\
& + a_j\sum_{k=1}^{3^{n-2}}\E(y^*_{i}|a_i< y^*_{i} <b_i, y^*_j \leq a_j, G_{k})P(G_{k}),\\
\end{split}
\label{app:cen_joint}
\end{equation}
where
\[
\begin{split}
&P(1)= P(y^*_i\leq a_i, y^*_j\geq b_j),P(3)= P(y^*_i\geq b_i, y^*_j\geq b_j),\\
& P(7)= P(y^*_i\leq a_i, y^*_j\leq a_j), P(9)= P(y^*_i\geq b_i, y^*_j\leq a_j),
\end{split} 
\]
and $ G_k $ for $ k=1,...,3^{n-2} $ denote a region (as in the case of the censored mean) where the multi-variable, $ \textbf{y}^*_{-i-j} $ lies on.

Concerning the last five terms of (\ref{app:cen_joint}), it is proved (as in case of second moment) that they depend only on the censoring limits 
\{$ a_i $, $ b_i $, $ a_j $, $ b_j $\}; thus, (\ref{app:cen_joint}) can be written as
\begin{equation}
\begin{split}
\E(y_{i}y_{j})&=a_ib_j P(1) +b_ib_j P(3) + a_ia_jP(7) + b_ia_jP(9)\\
& + b_j\E(y^*_{i}|a_i< y^*_{i} <b_i, y^*_j \geq b_j)P(2)\\
& + a_i\E(y^*_{j}|a_j< y^*_{j} <b_j, y^*_i \leq a_i)P(4)\\
& +\E(y^*_{i}y^*_{j}|a_i< y^*_{i} < b_i, a_j < y^*_{j} <b_j)P(5)\\
& + b_i\E(y^*_{j}|a_j< y^*_{j} <b_j, y^*_i \geq b_i)P(6)\\
& + a_j\E(y^*_{i}|a_i< y^*_{i} <b_i, y^*_j \leq a_j)P(8).\\
\end{split}
\label{cen_joint2}
\end{equation}
where
\[ 
\begin{split}
&P(2)= P(a_i <y^*_i< b_i, y^*_j\geq b_j),\\
&P(4)= P(y^*_i\leq a_i, a_j<y^*_j<b_j),\\
&P(5)= P(a_i<y^*_i<b_i, a_j<y^*_j< b_j),\\
&P(6)= P(y^*_i\geq b_i, a_j<y^*_j< b_j),\\
& P(8)= P(a_i<y^*_i< b_i, y^*_j\leq a_j).
\end{split} 
\]
At this point it should be noted that the truncated moments $ \E(y^*_{i}|\cdot) $ and $ \E(y^*_{i}y^*_{j}|\cdot)  $ in (\ref{cen_joint2}) are calculated by (\ref{trun_mean}) and  (\ref{trun_joint}), respectively. Although, the functions (\ref{F}), (\ref{F2}) in our case (censoring measurements) are defined only for the variables $ y^*_i $ and $y^*_j $, i.e.:
\[
F_i(x)=\frac{\int_{a_j}^{b_j}f_{Y^*_i,Y^*_j}(x,y^*_j)dy^*_{j}}{P(a_j<y^*_j<b_j, a_i<y^*_i<b_i)},
\]
and 
\[
F_{i,j}(x,y)=\frac{f_{Y^*_i,Y^*_j}(x,y)}{P(a_j<y^*_j<b_j, a_i<y^*_i<b_i)}.
\]
Therefore, the covariance matrix can be defined by the terms (\ref{app_mean}), (\ref{app_var}) and (\ref{cen_joint2}).

\section*{Appendix C: Evaluation of the Probabilities of the latent measurement to belong to the censored or uncensored region }

\label{Prob}
The mean of the latent measurement $ \textbf{y}^*_{k} $ given the saturated measurement $\textbf{y}_{k-1} $ is
\begin{equation}
\textbf{m}_k=\E(\textbf{H}\textbf{x}_k+\textbf{v}_k|\textbf{y}_{k-1})=\textbf{H}\E(\textbf{x}_k|\textbf{y}_{k-1})=\textbf{H}\hat{\textbf{x}}^-_{k}.
\label{eq:mean}
\end{equation}
The covariance matrix of $ \textbf{y}^*_{k}-\textbf{H}\hat{\textbf{x}}^-_{k} $ is
\[
\begin{split}
\textbf{S}_k=\Cov(\textbf{y}^*_{k}-\textbf{H}\hat{\textbf{x}}^-_{k})=\Cov(\textbf{H}\textbf{x}_k+\textbf{v}_k-\textbf{H}\hat{\textbf{x}}^-_{k}) =\Cov\big(\textbf{H}(\textbf{x}_k-\hat{\textbf{x}}^-_{k}))+\Cov(\textbf{v}_k)
\end{split}
\]
thus,
\begin{equation}
\textbf{S}_{k}=\textbf{H}\textbf{P}^-_{k}\textbf{H}^T+\textbf{R}.
\label{eq:cov}
\end{equation}
By (\ref{eq:mean}) and (\ref{eq:cov}) it is clear that $\textbf{y}^*_{k}|\textbf{y}_{k-1}\sim~N(\textbf{m}_k, \textbf{S}_k)$. The probability $ D^i_{\textbf{a},k} $ of the $ i^{th} $ component of the latent measurement $ \textbf{y}^*_k $ to be equal or less than $ a_i $ is
\[
\begin{split}
D^i_{\textbf{a},k}= P( y^*_{k,i}\leq a_i )= P\Bigg(\frac{y^*_{k,i}-m_{k,i}}{\sqrt{s_{(i,i),k}}}\leq \frac{a_i-m_{k,i}}{\sqrt{s_{(i,i),k}}} \Bigg)
\end{split}
\]
\begin{equation}
=\Phi\Bigg( \frac{a_i-m_{k,i}}{\sqrt{s_{(i,i),k}}}  \Bigg).
\end{equation}
Following the same procedure, the probability $ D^i_{\textbf{b},k} $ of the $ i^{th} $ component of the latent measurement $ \textbf{y}^*_k $ to be equal or bigger than $ b_i $ is
\begin{equation}
D^i_{\textbf{b},k}=1-\Phi\Bigg( \frac{b_i-m_{k,i}}{\sqrt{s_{(i,i),k}}} \Bigg).
\end{equation}
Finally, the probability of the $ i^{th} $ component of the latent measurement $ \textbf{y}^*_k $ to lie in the uncensored region $ (a_i, b_i) $ is
\begin{equation}
D^i_{un,k}=1-D^i_{\textbf{a},k}-D^i_{\textbf{b},k} .
\end{equation}

\bibliographystyle{unsrt} 
%\bibliography{mybib}

\end{document}